\documentclass[reqno,a4paper,10pt]{amsart}
\usepackage[a4paper, centering]{geometry}
\usepackage[utf8]{inputenc} % allow utf-8 input
\usepackage[english]{babel}
\usepackage[T1]{fontenc}
\usepackage{lmodern}
\usepackage{caption}
\usepackage{subcaption}
\usepackage{enumerate}
\usepackage{amssymb, amsmath}
\usepackage{mathrsfs,dsfont}
\usepackage{amscd}
\usepackage{bbm}
\usepackage{amsthm}
\usepackage{cite}
\usepackage{esint }
\usepackage[mathcal]{euscript}
\usepackage[active]{srcltx}
\usepackage{verbatim}
\usepackage[colorlinks,linkcolor={blue},citecolor={blue},urlcolor={black}]{hyperref}
\usepackage[]{changebar}
\usepackage{xcolor}
\usepackage{mathtools}
\usepackage{url}            % simple URL typesetting
\usepackage{booktabs}       % professional-quality tables
\usepackage{amsfonts}       % blackboard math symbols
\usepackage{nicefrac}       % compact symbols for 1/2, etc.
\usepackage{microtype}      % microtypography
\usepackage{lipsum}
\usepackage{fancyhdr}       % header
\usepackage{graphicx}       % graphics
\graphicspath{{media/}}     % organize your images and other figures under media/ folder
\usepackage{bookmark} %(fa vedere le formule)
\usepackage[a4paper, centering]{geometry}
\usepackage{epstopdf}
\usepackage{amscd}
\usepackage{color}
\usepackage{graphicx}

\usepackage{yfonts}
\usepackage{fullpage}
\usepackage{comment}
\usepackage[braket,qm]{qcircuit} % quantum circuits
\usepackage{ stmaryrd } % per un commando che chiamo lket dopo

\usepackage[many]{tcolorbox}
\newtcolorbox{mybox}{enhanced,colback=red!5!white, colframe=red!75!black, width=\textwidth,box align=center,halign=center,valign=center, center}
\usepackage[new]{old-arrows} % large arrows
\usepackage{mathabx} % updownarrows

\usepackage{tabularx}
\usepackage{multirow}
\usepackage{makecell}

\newtheorem{thm}{Theorem}[section]

\newtheorem*{thm*}{Theorem}
\newtheorem{cor}{Corollary}[section]

\newtheorem{lem}{Lemma}[section]

\newtheorem{prop}{Proposition}[section]

\newtheorem*{prop*}{Proposition}
\newtheorem{ass}{A}

\theoremstyle{definition}
\newtheorem{defn}{Definition}[section]

\theoremstyle{remark}
\newtheorem{rem}{Remark}[section]

\numberwithin{equation}{section}
\def\N{{\mathbb N}}

\def\R{{\mathbb R}}
\def\CC{{\mathbb C}}
\def\L{{\mathcal L}}
\def\Lg{{\mathscr{L}}}
\def\Lc{{\mathfrak L}}
\def\H{{\mathcal H}}
\def\O{{\mathcal O}}
\def\E{{\mathbb E}}
\def\dim{{\rm dim}}
\def\X{{\mathcal X}}
\def\Y{{\mathcal Y}}
\def\D{{\mathcal D}}
\def\tr{{\rm tr}}

\def\I{{\mathbbm{1}}}
\def\cP{{\mathscr{P}}}
\def\norma #1{\left\lVert #1 \right\rVert}

\def\P{{\mathbb{P}}}
\def\C{{\mathscr{C}}}
\def\de{{\rm d}}

\def\w{{\rm W}}

\def\K{{\mathcal K}}
\def\M{{\mathbb{M}}}
\def\nm #1{ \left\langle #1 \right\rangle}
\def\lin{{\rm lin}}
\def\ov #1{\overline{#1}}
\def\cE{{\mathcal{E}}}

\newcommand{\f}[1]{{\color{blue!85!black}#1}}
\newcommand{\id}{\mathbbm{1}}

\makeatletter
\def\cleardoublepage{\clearpage\if@twoside \ifodd\c@page\else
\hbox{}
\thispagestyle{empty}
\newpage
\if@twocolumn\hbox{}\newpage\fi\fi\fi}
\makeatother
\title[QML]{Quantitative convergence of trained quantum neural networks to a Gaussian process}

\author[A.~Melchor Hernandez]{Anderson Melchor Hernandez}
\address[A.~Melchor Hernandez]{Department of Mathematics, University of Bologna, Piazza di Porta San Donato 5, 40126, Bologna (Italy)}
\email{anderson.melchor@unibo.it}

\author[F.~Girardi]{Filippo Girardi}
\address[F.~Girardi]{Scuola Normale Superiore, piazza dei Cavalieri 7, 56126, Pisa (Italy), Korteweg-de Vries Institute for Mathematics, University of Amsterdam, Science Park, 105-107, 1098 XG, Amsterdam (The Netherlands)\\QuSoft, Science Park 123, 1098 XG Amsterdam (The Netherlands)}
\email{filippo.girardi@sns.it}

\author[D.~Pastorello]{Davide Pastorello}
\address[D.~Pastorello]{Department of Mathematics, University of Bologna, Piazza di Porta San Donato 5, 40126, Bologna (Italy)\\TIFPA-INFN, via Sommarive 14, 38123 Povo (Trento), Italy}
\email{davide.pastorello3@unibo.it}

\author[G.~De Palma]{Giacomo De Palma}
\address[G.~De Palma]{Department of Mathematics, University of Bologna, Piazza di Porta San Donato 5, 40126, Bologna (Italy)}
\email{giacomo.depalma@unibo.it}

\date{\today}

\begin{document}

\begin{abstract}
We study quantum neural networks where the generated function is the expectation value of the sum of single-qubit observables across all qubits. In [Girardi \emph{et al.}, arXiv:2402.08726], it is proven that the probability distributions of such generated functions converge in distribution to a Gaussian process in the limit of infinite width for both untrained networks with randomly initialized parameters and trained networks.
In this paper, we provide a quantitative proof of this convergence in terms of the Wasserstein distance of order $1$.
First, we establish an upper bound on the distance between the probability distribution of the function generated by any untrained network with finite width and the Gaussian process with the same covariance. This proof utilizes Stein’s method to estimate the Wasserstein distance of order $1$.
Next, we analyze the training dynamics of the network via gradient flow, proving an upper bound on the distance between the probability distribution of the function generated by the trained network and the corresponding Gaussian process. This proof is based on a quantitative upper bound on the maximum variation of a parameter during training. This bound implies that for sufficiently large widths, training occurs in the lazy regime, \emph{i.e.}, each parameter changes only by a small amount.
While the convergence result of [Girardi \emph{et al.}, arXiv:2402.08726] holds at a fixed training time, our upper bounds are uniform in time and hold even as $t \to \infty$.
\end{abstract}
\maketitle
\tableofcontents

\section{Introduction}\label{sec:1}

In recent years, scientific communities have become increasingly involved in the use of Artificial Intelligence (AI) to analyze large databases \cite{berlyand2023,russell2016}. AI currently encompasses a vast number of subfields, ranging from learning theory to the mathematical foundations of its development. The core of AI and machine learning lies in the recognition and identification of complex patterns within vast amounts of data, enabling these systems to uncover hidden relationships, make informed predictions, and generate insights that would be difficult, if not impossible, for humans to discern on their own \cite{bishop2006}. Among the new emerging disciplines, Quantum Machine Learning (QML) is an interdisciplinary field that merges the principles of quantum computing with classical machine learning techniques \cite{de2019primer,schuld2015, pastorello2023concise}. One of the core ideas of QML is to harness quantum algorithms and the unique properties of quantum mechanics such as superposition, entanglement, and quantum parallelism to enhance the performance of deep neural models \cite{biamonte2017}. Quantum neural networks constitute the quantum version of deep neural models. These new models are based on quantum circuits and generate functions given by the expectation values of a quantum observable measured on the output of a quantum circuit made by parametric one-qubit and two-qubit gates \cite{girardi2024,schuld2018}. The parameters of the circuit encode both the input data and the parameters of the model itself. These parameters are typically optimized by gradient descent, which involves iterative adjustment to minimize a cost function and improve the performance of the quantum circuit in the processing and analysis of data \cite{schuld2021effect}.

In this paper, we study quantum neural networks applied to supervised learning for binary classification. Let $\mathcal{X}$ be the set of possible inputs (\emph{e.g.}, pictures encoded as points in $\mathbb{R}^d$), which we assume to be finite.
Let $\Theta$ be the vector of the parameters, and let $x\mapsto f(\Theta,x)$ be the function generated by the quantum neural network.
Let $\left\{\left(x^{(i)},\,y^{(i)}\right):i=1,\,\ldots,\,n\right\}$ be the set of the training examples made by the training inputs $x^{(i)}\in\mathcal{X}$ (\emph{e.g.}, pictures of dogs or cats) and the corresponding training labels $y^{(i)}\in\left\{\pm1\right\}$ (\emph{e.g.}, $y^{(i)}=1$ if $x^{(i)}$ represents a dog and $y^{(i)}=-1$ if $x^{(i)}$ represents a cat).
The goal of supervised learning is to adjust the parameters $\Theta$ so that $f(\Theta,x)$ reproduces as closely as possible the training examples.
This is usually achieved by minimizing a loss function such as the empirical quadratic loss
\begin{equation}
\mathcal{L}(\Theta) = \sum_{i=1}^n\left(f(\Theta,x^{(i)}) - y^{(i)}\right)^2
\end{equation}
via gradient descent.
For simplicity, in this paper we will consider the continuous-time gradient flow rather than gradient descent.

Several works have focused on the analysis of quantum neural networks, as it is believed that they can combine the computational power of quantum computers with the capabilities of deep learning algorithms \cite{lloyd2020quantum}. In recent work \cite{liu2021}, the authors had shown that an exponential quantum speed-up can be obtained via the use of a quantum-enhanced feature space, where each data point is mapped in a non-linear way to a quantum state, and then classified by a linear classifier in a high-dimensional Hilbert space \cite{havlivcek2019}. Nevertheless, a significant disadvantage lies in the need to determine the appropriate parameters to configure the quantum circuit beforehand, and it is not yet clear whether these parameters can be effectively obtained using a variational technique \cite{cinelli2021var}. In fact, one of the primary challenges associated with training quantum neural networks is the so-called barren plateau phenomenon, related to different causes like entanglement in input states and the locality of observables, which poses a major obstacle by causing gradients to vanish during the optimization process, making it difficult to train the network effectively \cite{mcclean2018barren,shalev2017,larocca2024review}. To address this challenge, several methods have been proposed recently. For instance, in \cite{martin2023barren}, the authors consider the problem from the perspective of quantum tensor network optimization;
in \cite{ragone2024lie}, the authors presented a general theory based on the Lie algebra of the circuit's generators to describe the rise of barren plateaus; in \cite{park2024hamiltonian}, the authors solved the problem of barren plateaus in the specific case of the Hamiltonian variational ansatz, a kind of quantum circuit used in quantum many-body problems; in \cite{liu2024mitigating}, an efficient state ansatz is proposed to mitigate barren plateaus in the context of variational quantum eigensolvers; in \cite{zhang2024absence}, a lower bound on the variance of the circuit gradients is derived for quantum circuits composed by local 2-designs.

In \cite{cerezo2021}, the authors considered an untrained variational quantum circuit and demonstrated that the associated cost function exhibits an exponentially vanishing gradient, highlighting the difficulty of overcoming this issue. In contrast, the authors in \cite{kiani2022} explored more general cost functions inspired by the quantum generalization of optimal mass transport theory \cite{villani2009optimal} developed in \cite{de2021quantum}, offering a different approach to mitigate the challenges posed by the barren-plateau phenomenon. An important problem in the classical theory of deep learning is the analysis of the behavior of a neural network in the limit of infinite width. Significant progress has been made in addressing the question of whether training can perfectly fit the training examples while simultaneously avoiding overfitting. A fundamental breakthrough has been the proof that, in the limit of infinite width, the probability distribution of the function generated by a deep neural network trained on a supervised learning problem converges to a Gaussian process \cite{hanin2018neural,lee2017deep,lee2019wide}. These recent developments have inspired a renewed interest in this important result within the field of quantum machine learning. Specifically, they have led to investigations into whether quantum neural networks exhibit similar properties. In this context, several studies have emerged. For example, \cite{abedi2023} proves that in the infinite width limit, trained quantum neural networks with constant depth operate in the lazy regime (\emph{i.e.}, the maximum amount by which the training can change a parameter tends to zero) and are capable of perfectly fitting the training examples. For a contrasting scenario, see \cite{garcia2023deep}.

In the recent work \cite{girardi2024}, the authors rigorously generalized the above classical breakthrough to the context of quantum neural networks. They considered quantum neural networks trained on supervised learning tasks, where the objective function is defined as the expected value of the sum of single-qubit observables across all qubits. The authors proved, for the first time, the trainability in the limit of infinite width in any regime where the depth is allowed to grow with the number of qubits (denoted by $m$), as long as barren plateaus do not arise.
More precisely, Ref. \cite{girardi2024} first proves that the probability distribution of the function generated by a randomly initialized quantum neural network in the limit of infinite width converges in distribution to a Gaussian process when the parameters on which each measured qubit depends influence a small number of other measured qubits (\cite[Theorem 3.14]{girardi2024}). Ref. \cite{girardi2024} then proves that for quantum neural networks trained in continuous time with gradient flow, the training occurs in the lazy regime and is able to perfectly fit the training set. The key element of the proof is showing that the dependence of the generated function on the parameters can be approximated by its linearized version near the initialization values of the parameters. Consequently, the linearized evolution equation has an analytic solution whose probability distribution is Gaussian with analytically computable mean and covariance. As a result, the probability distribution of the function generated by the trained network is proved to converge in distribution to the aforementioned Gaussian process (\cite[Theorem 4.15]{girardi2024}).

\subsection{Our results}
In this paper, we prove a quantitative version of the results of \cite{girardi2024}. Our first goal is to establish explicit upper bounds on the distance between the probability distribution of the function generated by the quantum network at initialization and the corresponding limit Gaussian process. We employ the Wasserstein distance of order $1$, which we denote with $\mathrm{d_W}$, and we provide an explicit upper bound that tends to zero as the number of qubits diverges. Without giving all the details, we prove the following (see \autoref{thmwass1} for the formal statement):

\begin{thm}[Convergence at initialization, informal statement]\label{intro:thm1}
We denote with $\ov{X}$ the vector made by the elements of $\mathcal{X}$, and with $f(\Theta,\ov{X})$ the vector made by the associated outputs.
Let $\ov{\K}_{0}$ be the covariance matrix of $f(\Theta,\ov{X})$ when the parameters $\Theta$ are randomly initialized. Then, there exist positive numbers $\left\{\alpha_{m}^{\ov{N},\ov{\K}_{0}}\right\}$ given by \eqref{ratewasst1} and depending on the number of qubits of the network $m$, on the architecture of the network, on the number of possible inputs $\ov{N}$ and on the covariance matrix $\ov{\K}_{0}$ such that
\begin{align}\label{intro:ratewasst1}
\de_{\w}\left(f(\Theta,\ov{X}),\mathcal{N}(0,\ov{\K}_{0})\right)\leq \alpha_{m}^{\ov{N},\ov{\K}_{0}},
\end{align}
where $\mathcal{N}(0,\ov{\K}_{0})$ denotes the centered Gaussian distribution with covariance matrix $\ov{\K}_{0}$.    
\end{thm}
To obtain the explicit expression \eqref{ratewasst1} for $\alpha_{m}^{\ov{N},\ov{\K}_{0}}$, we employ Stein's method, which is a powerful method to estimate the Wasserstein distance of order $1$ of the sum of weakly dependent random variables \cite{nourdin2012,nourdin2010,stein1986approximate,y2015stein,gall2018}. We will prove that $\alpha_{m}^{\ov{N},\ov{\K}_{0}}$ depends on the size of the light cones of the quantum circuit and on a normalizing constant that quantifies the presence or absence of the barren-plateau phenomenon.
The dependence of $\alpha_{m}^{\ov{N},\ov{\K}_{0}}$ on $\ov{N}$ scales as $\ov{N}^{\frac{3}{2}}$, causing $\alpha_{m}^{\ov{N},\ov{\K}_{0}}$ to diverge as the number of inputs increases: the nature of the network introduces dependencies between different qubits, which cause the bound \eqref{intro:ratewasst1} to be worse than the bound obtained for the sum of independent random variables in the central limit theorem.

Then, we quantify the distance between the probability distribution of the function generated by the trained network and its associated Gaussian process. We prove a strong quantitative version of the results of \cite{girardi2024} that can be applied to any quantum neural network with finite width and therefore does not require to build a sequence of networks of increasing width.
As before, without giving all the details, we prove the following (see \autoref{thmwass4} for the formal statement):

\begin{thm}[Convergence of the trained network, informal statement]\label{intro:thm2}
With the same notation of \autoref{intro:thm1}, let us denote by $f(\Theta_{t},\ov{X})$ the vector of the outputs of a quantum neural network trained via gradient flow for a time $t>0$, where $\Theta_{t}$ is the vector of the trained parameters and $\ov{X}$ is the vector of the possible inputs.
Let $X$ be the vector of the training inputs. Then, there exists a Gaussian probability distribution $\mathcal{N}\left(\mu_{t}(\ov{X}),\K_{t}(\ov{X},\ov{X})\right)$ with mean $\mu_{t}\coloneqq \mu_{t}(\ov{X})$ and covariance matrix $\ov{\K}_{t}\coloneqq \K_{t}(\ov{X},\ov{X})$ (that will be determined analytically in \eqref{newmedia} and \eqref{covlimit}) such that
\begin{align}\label{intro:ratewasst2}
\de_{\w}^{(s)}\left(f(\Theta_{t},\ov{X}),\mathcal{N}(\mu_{t},\ov{\K}_{t})\right)\leq\gamma_{m}^{\ov{N},n,\ov{\K}_{0}},
\end{align}
where $\left\{\gamma_{m}^{\ov{N},n,\ov{\K}_{0}}\right\}$ are positive numbers given by \eqref{otherrate4} depending on the number of possible inputs $\ov{N}$, on the number of qubits $m$, on the architecture of the network, on the covariance matrix at initialization $\ov{\K}_{0}$ and on the number of training examples $n$, and $\de_{\w}^{(s)}$ denotes the truncated Wasserstein distance of order $1$ defined in \eqref{def:swassers}.
\end{thm}
Let us point out that $\gamma_{m}^{\ov{N},n,\ov{\K}_{0}}$ depends on the previous constant $\alpha_{m}^{\ov{N},\ov{\K}_{0}}$, and thus its dependence on $\ov{N}$ is of order $\ov{N}^{\frac{3}{2}}$.
We stress that \eqref{intro:ratewasst2} is valid for any quantum neural network of finite width, unlike the results of \cite{girardi2024} which require to build a sequence of quantum neural networks with diverging width.
Furthermore, \cite{girardi2024} proved the convergence to a Gaussian process at fixed training time, and such result cannot be directly extended to the limit $t\to\infty$, unless such limit is taken after the limit of infinite width.
On the contrary, the constants $\left\{\gamma_{m}^{\ov{N},n,\ov{\K}_{0}}\right\}$ do not depend on the training time, and therefore the bound \eqref{intro:ratewasst2} is uniform in time and holds for finite width even in the limit $t\to\infty$.
Lastly, let us notice that \eqref{intro:ratewasst1} and \eqref{intro:ratewasst2} imply the convergence results of \cite{girardi2024}.

As in \cite{girardi2024}, the proof of \autoref{intro:thm2} is based on the following upper bound to the maximum variation of a parameter during training (see \autoref{newgradfl} for the formal statement):
\begin{thm}[Lazy training, informal statement]\label{intro:thm3}
With the same notation of \autoref{intro:thm1} and \autoref{intro:thm2}, let us denote by $x\mapsto f^{\mathrm{lin}}(\Theta,x)$ the first-order Taylor approximation of $x\mapsto f(\Theta,x)$ with respect to the parameters $\Theta$ expanded around their initialization values. Let $x\mapsto f(\Theta_t^{\mathrm{lin}},x)$ be the model obtained by randomly initializing $\Theta$ and training $x\mapsto f^{\mathrm{lin}}(\Theta,x)$ via gradient flow for time $t$.
Then, for any $0<\delta<1$ there exist positive numbers $\left\{\eta_{m,n,\delta}\right\}$ given by \eqref{grad3} depending on $m,n,\delta$, and further constants, such that with probability at least $1-\delta$, one gets that
\begin{align}\label{intro:ratewasst3}
\sup_{\substack{x\in\mathcal{X}\\t\geq 0}}|f(\Theta_t,x)-f^{\mathrm{lin}}(\Theta_t^{\mathrm{lin}},x)|&\leq \eta_{m,n,\delta},
\end{align}
where $\mathcal{X}$ is the set of all the possible inputs.
\end{thm}
We stress that, for any fixed $n$ and $\delta$, $\eta_{m,n,\delta}\to0$ for $m\to\infty$ if barren plateaus do not arise.
Therefore, \eqref{intro:ratewasst3} proves that if $m$ is large enough, each parameter remains with high probability close to its initialization value and that the training happens in the lazy regime. In the formal statement of \autoref{intro:thm3}, we provide two main improvements with respect to \cite[Theorems 4.17, 4.20]{girardi2024}.
First, \cite[Theorems 4.17, 4.20]{girardi2024} require to build a sequence of networks with increasing width and hold only if the width is larger than some value that has not been computed.
Instead, we want to describe the training dynamics of a network with finite width, therefore we have replaced the hypotheses about the asymptotic behavior of the network with the quantitative requirement at finite size \eqref{hp:lambda_pos}.
Second, while the bounds of \cite[Theorems 4.17, 4.20]{girardi2024} contain constants that have not been computed and that depend on the failure probability of the training at finite width, we have provided explicit bounds.\\

The article is organized as follows.
In \autoref{sec:2}, we introduce some preliminaries, notations, and our hypotheses. In \autoref{sec:gausproc}, we recall some known facts about Gaussian processes. In \autoref{Sec:Stein}, we recall some facts about the Stein method and its relationship to the Wasserstein distance of order $1$. In \autoref{Sec:results}, we state our main results. In \autoref{sec:init}, we prove \autoref{intro:thm1}. In \autoref{sec:train}, we prove \autoref{intro:thm3}, which is then used to prove \autoref{intro:thm2}. In the same section, we prove \autoref{intro:thm2}. We conclude in \autoref{sec:concl}. Lastly, for the sake of completeness, we include appendices to discuss the regularity of the solutions of the Stein's equation (\autoref{sec:regularity}) and some theorems proved in \cite{girardi2024} regarding the convergence of the function generated by a quantum circuit to a Gaussian process  (\autoref{sec:towgauss}).

\begin{table}[t]
    \caption{Notation concerning the general properties of the circuit and the Wasserstein distance.}
    \label{table1}
    \begin{tabularx}{\textwidth}{p{0.1\textwidth}X>{\raggedleft\arraybackslash}l} 
    \toprule
    Symbol & Description & Introduced in \\
    \midrule
      %{\underline{Indices}} \\
      $m$ & number of qubits in the parameterized quantum circuit & \autoref{sub:quantum}\\ 
      $L$ & number of layers in the parameterized quantum circuit& Def. \ref{def:numbL}\\
      %$[\ell \, m]$ & layer-qubit representation & Def. \ref{def:layerq}\\
      $\Theta$ & vector of the parameters of the quantum circuit&\autoref{parm}\\
      $\mathscr{P}$& denotes the parameter space, so that $\Theta\in \mathscr{P}$. Here $\mathscr{P}=[0,\pi]^{Lm}$& \autoref{sub:trdata}\\
      $|\Theta|$ & number of parameters $|\Theta|:=\dim\mathscr{P}=Lm$ & \autoref{sub:quantum} \\
      $U(\Theta,x)$ & parameterized quantum circuit (unitary operator) & \textit{ibid.}\\
      $f(\Theta,x)$ & function generated by the quantum neural network & \autoref{sub:assumpt}\\
      $N(m)$ & normalization factor of the model & \textit{ibid.}\\
      $\mathcal{M}_i$& (extended) future light cone of the parameter $i$ & Def. \ref{extcone2}\\
      $\mathcal{N}_k$ &(extended) past light cone of the observable $k$& Def. \ref{lcone1}\\
      $|\mathcal{M}|$ & maximal cardinality of a future light cone in the circuit &  \autoref{maxicard}\\
      $|\mathcal{N}|$ & maximal cardinality of a past light cone in the circuit & \textit{ibid.}\\
      $\mathcal{P}_{i}$ &set of indices of the observables depending on the observable $i$&\autoref{pk1}\\
      $\widetilde{\mathcal{P}}_{i}$& set representing the union of the sets $\mathcal{P}_{j}$ for those $j$ in $\mathcal{P}_{i}$ &  \autoref{newsets}\\
      $\mathrm{d_W}$ & Wasserstein distance of order $1$ & \autoref{sub:distances}\\
      $\mathrm{d}_\mathrm{W}^{(s)}$ & truncated Wasserstein distance of order $1$ & \textit{ibid}\\
      $\omega$ & continuity modulus of the Hessian of the solution of Stein's equation & \autoref{contfunct}\\
      \bottomrule
    \end{tabularx}
\end{table}
\begin{table}[t]
    \caption{Notation concerning the training of the circuit.}
    \label{table2}
    \begin{tabularx}{\textwidth}{p{0.1\textwidth}X>{\raggedleft\arraybackslash}l} 
    \toprule
    Symbol & Description & Introduced in \\
    \midrule
      $\mathcal{X}$ & the feature space & \autoref{sub:trdata}\\
       $\ov{N}$ &  cardinality of the feature space& \textit{ibid}\\
       $\ov{X}$ &  vector containing all the possible inputs in $\mathcal{X}$ & \textit{ibid}\\
      $x$ & a generic input belonging to $\mathcal{X}$& \textit{ibid}\\
      $\mathcal{Y}$ & the output space & \textit{ibid}\\
      $y$ & a generic output belonging to $\mathcal{Y}$& \textit{ibid}\\
      $\mathcal{D}$ & training set, whose elements are denoted by $(x^{(i)},y^{(i)})$ for $i=1,\dots,n$ & \textit{ibid.} \\
      $n$ & number of training samples (i.e., cardinality of $\mathcal{D}$) & \textit{ibid.}\\
      $X$ & vector containing the inputs of the training set & \autoref{sub:NTK}\\
      $Y$ & vector containing the outputs of the training set corresponding to $X$ & \textit{ibid.}\\
      $f^{\mathrm{lin}}(\Theta,x)$ & linearized model & \autoref{sub:linearmod}\\
      $ \hat K_{\Theta}(x,x')$ & empirical neural tangent kernel & Def. \ref{def:ENTK}\\
      $ K(x,x')$ & analytic neural tangent kernel & Def. \ref{def:analk1}\\
      $\ov{\K}_{0}$& Covariance matrix at initialization of size $\ov{N}\times \ov{N}$& \autoref{thmwass1}\\
      $\lambda_{\min}^K$ & smallest eigenvalue of $K(X,X^T)$ & Assumption \autoref{A3} \\
      %$ \ov{K}(x,x')$ & limit kernel of the rescaled analytic neural tangent kernel & \autoref{limitk}\\
      $t$ & continuous or discrete training time & \autoref{sub:NTK}\\
      $F(t)$ & vector containing the model function evaluated in the inputs of the dataset, i.e., $F(t)=f(\Theta_t,X)$ & \autoref{fvectmod}\\
      $\ov{\K}_{t}$& Covariance matrix of $f(\Theta_{t},\ov{X})$& \autoref{thmwass4}\\
      $F^{\mathrm{lin}}(t)$ & vector containing the linearized model function evaluated in the inputs of the dataset, i.e., $F^{\mathrm{lin}}(t)=f^{\mathrm{lin}}(\Theta^{\mathrm{lin}}_t,X)$ & \autoref{linevol1} \\
      $\eta$ & learning rate, which enters the gradient flow equation and is a function of $m$ & \autoref{gradform1}\\
      $\mathcal{L}(\Theta)$ & cost function for the original model according to the training set & \autoref{costfunct2}\\
      $\Theta_t$ & parameter vector evolving via gradient flow according to $\mathcal{L}$ & \autoref{gradform1} \\
      $\mathcal{L}^{\mathrm{lin}}(\Theta)$ & cost function for the linearized model & \autoref{sub:linearmod} \\
      $\Theta_t^{\mathrm{lin}}$ & parameter vector evolving via gradient flow according to $\mathcal{L}^{\lin}$ & \textit{ibid.}\\
      \bottomrule
    \end{tabularx}
\end{table}

\section{Preliminaries}\label{sec:2}
Let us start by introducing the notation of the present work.
\subsection{Training data}\label{sub:trdata}
Let $\X\subset \R^{d}$ be the feature space, \emph{i.e.}, the set of all the possible inputs, which we assume to have finite cardinality $\ov{N} = |\mathcal{X}|$. We will often use the notation $\bar X$ to represent the vector having as entries the elements of the input space $\mathcal{X}$. Let $\mathbb{R}$ be the output space.
Let
\begin{equation}
    \D \coloneqq\left\{(x^{(i)},y^{(i)}):i=1,\ldots,n\right\}\subset \X\times\Y
\end{equation}
be the training set. We set $n=\vert \D\vert$ to be the cardinality of $\D$. We let $\cP$ be the parameter space, and let $\Theta\in\cP$ be the vector of the parameters. Let $f:\cP\times \X\rightarrow \mathbb{R}$ be a generic  parametric function.
As a cost function, we consider the mean squared error on the training set $\D$ of cardinality $n$
\begin{align}\label{costfunct2}
 \L(\Theta)\coloneqq \sum_{i
=1}^{n}\left(f(\Theta,x^{(i)})-y^{(i)}\right)^{2}.  
\end{align}
\subsection{Quantum neural networks}\label{sub:quantum}
Let $\CC^{2}$ be the Hilbert space of a single qubit. In what follows, we denote by $m\in \N$ the number of qubits of the quantum neural network. Hence, the Hilbert space of the system is $\H=\left(\CC^{2}\right)^{\otimes m}$, and its dimension denoted as $\dim\,\H$ is $2^{m}$.  Following the notations of \cite{girardi2024}, we recall what a ``layer'' is.
\begin{defn}\label{def:numbL}
A layer is a unitary operation $U(\Theta,x)\in \Lc(\H)$ resulting from:
\begin{enumerate}
    \item[$1.$] the application on each qubit of a different parametrized single-qubit gate $W_{i}(\Theta)\in \Lc(\CC^{2})$; each parametrized gate depends on a single parameter $\theta_{i}$, which is different for each gate,
\end{enumerate}
followed by
\begin{enumerate}
    \item[$2.$] a set of one-qubit and two-qubit gates acting on disjoint qubits, that is, each qubit can be acted at most one gate; each gate may depend only on the input $x$; the resulting unitary operation will be called $V\in \Lc(\H)$. 
\end{enumerate}
In what follows, a quantum circuit is a combination of parameterized layers $U_{\ell}(\Theta,x)$, $\ell\in \mathbb{N}$. Next, we let $L\in \mathbb{N}$ be the number of layers in a quantum circuit, which may depend on  the number of qubits $m$.
\end{defn}
In the next, we consider $\theta_{1},\ldots, \theta_{Lm}$ be the parameters of a quantum circuit,  so that  $\Theta$ will be the vector

\begin{align}\label{parm}
\Theta\coloneqq
 \begin{pmatrix}
 \theta_{1}\\
 \theta_{2}\\
 \vdots\\
 \theta_{Lm}
 \end{pmatrix},
\end{align}
of dimension $\mathrm{dim}\Theta\coloneqq\vert \Theta\vert=Lm$. We now recall a convenient notation for the indices of the parameters used in \cite{girardi2024}.

\begin{defn}\label{def:layerq}
Each parameter index $i\in \{1,\ldots, Lm\}$ can be expressed in the form $i=m(\ell-1)+k$ for some $\ell\in\{1,\ldots,L\}$, and $k\in\{1,\ldots,m\}$.  Here, $k$ refers to the qubit involved in the single-qubit gate parametrized by $\theta_{i}$, while $\ell$ refers to the layer in which such gate acts. The following compact notation, which we call layer-qubit representation of the parameter index $i$, simplifies the above form:
\begin{align}
    i=[\ell m]\equiv m(\ell-1)+k.
\end{align}
\end{defn}
Therefore, a layer $U_{\ell}(\Theta,x)$ can be written as

\begin{align}
\begin{aligned}
U_{\ell}(\Theta,x)&\coloneqq V_{\ell}(x)\left(W_{[\ell 1]}\otimes \cdots \otimes W_{[\ell m]}\right)(\Theta)\\
&=V_{\ell}(x)W_{\ell}(\Theta),
\end{aligned}
\end{align}
where we have set $W_{\ell}(\Theta)\coloneqq \left(W_{[\ell 1]}\otimes \cdots \otimes W_{[\ell m]}\right)(\Theta)$. The result of the circuit on a initial state $\ket{\psi_{0}}$ is described by the unitary operation 

\begin{align}\label{formula:1}
U(\Theta,x)\coloneqq U_{L}(\Theta,x)\cdots U_{1}(\Theta,x); \qquad \ket{\psi_{{\rm out}}}\coloneqq U(\Theta,x)\ket{\psi_{0}}.
\end{align}

\subsection{Light cones}\label{sub:lightc}
The architecture of the network generates a causal structure where the probability distribution of the outcome of the measurement of each output qubit can depend only on some of the parameters, and each parameter can influence only some output qubits.
Such causal structure is formalized by the notion of light cones:
\begin{defn}[Light cones]\label{lightcones}
For any $i\in\{1,\ldots, \vert \Theta\vert\}$, we define the future light cone $\Lg_{i}^{f}$ of the parameter $\theta_{i}$ as the subset 

\begin{align}\label{lcone1}
\Lg_{i}^{f}\coloneqq \left\{k\in\{1,\ldots,m\}:\text{$f_{k}(\Theta,x)$ depends on $\theta_{i}$}\right\}.
\end{align}
Analogously, we define the past light cone $\Lg_{k}^{p}$ of the qubit $k$ as the subset
\begin{align}\label{lcone2}
\Lg_{k}^{p}\coloneqq\{i\in\{1,\ldots,\vert \Theta\vert\}: \text{$f_{k}(\Theta,x)$ depends on $\theta_{i}$}\}.
\end{align}
\end{defn}
Both sets $\Lg_{i}^{f}$, and $\Lg_{k}^{f}$ are useful for tracking the dependence of observables on the parameters. In general, it is difficult to provide an explicit representation of them. For this reason, we now introduce another family of sets that can help us to explicitly track the dependence on the parameters. For any quantum circuit $U$, we define the following sets. For each layer $\ell$, and qubit $k$, we set

\begin{align}\label{auxset1}
    \mathcal{I}_{\ell,k}\coloneqq \{k'\in\{1,\ldots,m\}:\text{the qubit $k$ interacts with the qubit $k'$ in the layer $\ell$}\}\cup \{k\}.
\end{align}
We now set,

\begin{align}\label{auxset2}
\mathcal{J}_{k}^{\ell}\coloneqq 
\begin{cases}
& \mathcal{I}_{L,k} \hskip 0,3cm \text{if $\ell=L$,}\\
&\displaystyle\bigcup_{k'\in \mathcal{J}_{k}^{\ell+1}}\mathcal{I}_{\ell,k'}\hskip 0,3cm \text{if $\ell<L$}.
\end{cases}
\end{align}

In particular $\mathcal{J}_k^1$ is the set of qubits in the past light cone of the observable $k$, i.e., the qubits involved in the computation of its expectation value.

Furthermore, we set

\begin{align}
\mathcal{N}_{k}^{\ell}\coloneqq \displaystyle\bigcup_{k'\in \mathcal{J}_{k}^{\ell}}\{[\ell k']\}.
\end{align}

\begin{defn}[Extended light cones]\label{extlightc}
Let us fix a quantum circuit $U$. Given any qubit index $k\in\{1,\ldots,m\}$, we define the extended past light cone $\mathcal{N}_{k}$ as the subset of the parameter indices $\{1,\ldots, \vert \Theta\vert\}$ given by

\begin{align}\label{extcone1}
   \mathcal{N}_{k}\coloneqq \displaystyle \bigcup_{\ell=1}^{L}\mathcal{N}_{k}^{\ell} 
\end{align}
\end{defn}
Similarly, we define the extended future light cone of a parameter index $i\in\{1,\ldots, \vert \Theta\vert\}$, as

\begin{align}\label{extcone2}
 \mathcal{M}_{i}\coloneqq \{k\in\{1,\ldots,m\}: i\in \mathcal{N}_{k}\}. 
\end{align}
In the next, we set 
\begin{align}\label{maxicard}
\begin{aligned}
&|\mathcal{M}|\coloneqq \displaystyle\max_{i}\vert \mathcal{M}_{i}\vert;
&|\mathcal{N}|\coloneqq \displaystyle \max_{k}\vert \mathcal{N}_{k}\vert
\end{aligned}
\end{align}
the maximal cardinalities of the extended light cones.
\begin{rem}
We notice that $\Lg_{k}^{p}\subset \mathcal{N}_{k}$, and $\Lg_{i}^{f}\subset \mathcal{M}_{i}$. and thus from now on, we can only consider the extended light cones. 
\end{rem}
\subsection{Assumptions on architecture and initialization}\label{sub:assumpt}
We assume the following:
\begin{ass}\label{A1}
We consider an observable $\O$ given by the sum of single qubit observables $\O_{k}$:
    \begin{equation}\label{ipot1}
     \O=\sum_{k=1}^{m}\O_{k} = \O_{1}\otimes \I_{2}\otimes\cdots\otimes \I_{m} +\I_{1}\otimes \O_{2}\otimes\cdots\otimes \I_{m} +\I_{1}\otimes \I_{2}\otimes\cdots\otimes \O_{m},
    \end{equation}
    where each $\O_k$ is traceless and has the spectrum contained in the interval $[-1,1]$. We further assume that the parametric one-qubit gates of the circuit $W_{i}(\theta_{i})$ can be written as time evolutions generated by hermitian hamiltonians $\mathcal{G}_{i}$ with spectrum in $\{-1,1\}$, \emph{i.e.}, 
    \begin{align}\label{shape1}
        W_{i}(\theta_{i})=e^{-i\mathcal{G}_{i}\theta_{i}}\,,\qquad \mathcal{G}_i = \mathcal{G}_i^\dag = \mathcal{G}_i^{-1}\,.
    \end{align}
We notice that $\theta_i\mapsto W_{i}(\theta_{i})$ is periodic with period $\pi$ up to an irrelevant multiplicative constant, so we fix the parameter space to be $\mathscr{P}=[0,\pi]^{Lm}$, and thus $\vert \Theta \vert={\mathrm dim}\mathscr{P}=Lm$.
The function generated by the network is then
    \begin{align}\label{model1}
    \begin{aligned}
        f(\Theta,x)&\coloneqq \frac{1}{N(m)}\bra{0^{m}}U^{\dag}(\Theta,x)\O\,U(\Theta,x)\ket{0^{m}}\\
        &=\frac{1}{N(m)}\sum_{k=1}^{m}f_{k}(\Theta,x)
        \end{aligned}
    \end{align}
    where
    \begin{align}\label{model2}
    f_{k}(\Theta,x)\coloneqq\bra{0^{m}}U^{\dag}(\Theta,x)\O_{k}\,U(\Theta,x)\ket{0^{m}}.
    \end{align}
\end{ass}

\begin{ass}\label{A2}
We assume that each parameter $\theta_i$ is initialized independently by sampling it from the uniform distribution on $[0,\pi]$ and that the final layer is chosen so that $\E_\Theta\,f_{k}(\Theta,x)=0$ for any $k=1,\ldots,m$ and any $x\in \X$.
    For any $x,\,x'\in\mathcal{X}$, let
    \begin{equation}
        \mathcal{K}_0(x,x') = \mathbb{E}\left(f(\Theta,x)\,f(\Theta,x')\right)\,.
    \end{equation}
    We assume that $\mathcal{K}_0(x,x)>0$ for any $x\in\mathcal{X}$ and that $N(m)$ is chosen such that
\begin{align}\label{ipot2new}
    \max_{x\in\mathcal{X}} \mathcal{K}_0(x,x) = 1\,.
\end{align}
    Finally, we assume that $N(m)\geq 1$ (our bounds would anyway become trivial if $N(m)\leq 1$).
\end{ass}
In what follows, let us set 

\begin{align}
V\coloneqq \{1,\ldots,m\}.
\end{align}
For each $k\in V$, we define

\begin{align}\label{pk1}
\mathcal{P}_{k}\coloneqq \{k'\in V: \text{$f_{k'}(\Theta,x)$ is not independent from $f_{k}(\Theta,x)$}\}.   
\end{align}
This set is crucial since takes track of the number of random variables $f_{k'}(\Theta,x)$ that have correlation with $f_{k}(\Theta,x)$. Let $\mathscr{G}=(V,E)$ be the graph  with vertices $V$, and edges $E$ defined as follows. We say
\begin{align}\label{graphrel}
\text{$(k,k')\in E$ if and only if $k'\in \mathcal{P}_{k}$.}   
\end{align}

Furthermore, we define the {\em maximal degree} $D$ of $\mathscr{G}$ as the maximum number of edges containing any fixed vertex as
\begin{align}\label{maxdegree}
&D\coloneqq \max_{k\in V}{\rm deg}\,k=\max_{k\in V}\vert\{k'\in V: (k,k')\in E\} \vert=\max_{k\in V}|\mathcal{P}_k|.   
\end{align}
Let us notice that according to the definition of $\mathcal{P}_{i}$, we have that $j\in \mathcal{P}_{i}$ if and only if $i\in \mathcal{P}_{j}$.
Let ${\rm dist}$ be the distance on $\mathscr{G}$ given by the length of the shortest path, such that for any $i\in V$ we have
\begin{equation}
\mathcal{P}_i = \left\{j\in V : \mathrm{dist}(i,j) \le 1\right\}\,.
\end{equation}

In what follows, we set
\begin{align}\label{newsets}
\widetilde{\mathcal{P}}_{i}\coloneqq \displaystyle\bigcup_{j\in \mathcal{P}_{i}}\mathcal{P}_{j} = \left\{j\in V : \mathrm{dist}(i,j) \le 2\right\},   
\end{align}
and 

\begin{align}\label{indmeas}
&\widetilde{D}\coloneqq \max_{1\leq i\leq m}\left\vert \left\{ (i,j): \widetilde{\mathcal{P}}_{i}\cap\widetilde{\mathcal{P}}_{j}\neq \emptyset\right\}\right\vert = \max_{1\leq i\leq m}\left\vert \left\{ (i,j): \mathrm{dist}(i,j)\le 4 \right\}\right\vert.
\end{align}

\begin{lem}[{\cite[Lemma 2.25]{girardi2024}}]\label{stimacone1}
 For any $k\in\{1,\ldots,m\}$, let $\mathcal{P}_{k}$ be defined as in \eqref{pk1}. Then 
 \begin{align}\label{conseqlem}
   \vert\mathcal{P}_{k} \vert\leq\vert\mathcal{M} \vert \vert \mathcal{N}\vert.
 \end{align}
 In particular,
 \begin{align}
   D\leq\vert\mathcal{M} \vert \vert \mathcal{N}\vert.
 \end{align}
\end{lem}
\begin{lem}\label{lem:stimatildeD}
We have
\begin{equation}
    \widetilde{D} \le \vert\mathcal{M} \vert^4\, \vert \mathcal{N}\vert^4\,.
\end{equation}
\end{lem}

\begin{proof}
Note that by definition of $D$, we have that
\begin{align*}
\widetilde{D}\le D^4,
\end{align*}
and so that by \autoref{stimacone1}, we are done.
\end{proof}

\begin{rem}\label{usedremark1}
For any $i=1,\,\ldots m$, let
\begin{align}\label{eq:defZi}
Z_{i}\coloneqq \sum_{k\in \mathcal{P}_{i}}\frac{f_{k}(\Theta,\ov{X})}{N(m)}.
\end{align}
Then if $\mathcal{P}_{k}\cap \widetilde{\mathcal{P}}_{j}=\emptyset$, we have that $Z_{j}$ is independent of $Z_{k}$. This fact will be significant when proving \autoref{thmwass1}.
\end{rem}

\subsection{The neural tangent kernel}\label{sub:NTK}
Before presenting our main result, let us review some relevant facts about the quantum neural tangent kernel as presented in \cite{girardi2024}. We are interested in the analysis of the minimization of the cost function \eqref{costfunct2} via gradient flow:
\begin{align}\label{gradform1}
\frac{\de \Theta_{t}}{\de\,t}=-\eta \nabla_{\Theta}\L(\Theta_{t})  
\end{align}
where $\eta>0$ is the learning rate that can be reabsorbed by rescaling the training time, and the initial value of $\Theta$ is given by the random sampling of \autoref{A2}. Notice that 
\begin{align}\label{gradvar1}
\frac{\de }{\de t}\L(\Theta_{t})=\frac{\de \Theta_{t}}{\de t}\cdot \nabla_{\Theta}\L(\Theta_{t})=-\eta \norma{\nabla_{\Theta}\L(\Theta_{t})}_{2}^{2}\leq 0.   
\end{align}
We stress that in general, the loss function $\mathcal{L}(\Theta)$ is not convex, hence gradient flow is not guaranteed to converge to a global minimum.
Given a training set $\mathcal{D}=\{(x^{(i)},y^{(i)})\}_{i=1,\dots,n}\subseteq \mathcal{X}\times\mathcal{Y}$, we will call $n=|\mathcal{D}|$ the number of examples and we will represent it in a vectorized form as follows
\begin{align}
X=\begin{pmatrix} x^{(1)}\\x^{(2)}\\\vdots\\x^{(n)} \end{pmatrix},\qquad 
Y=\begin{pmatrix} y^{(1)}\\y^{(2)}\\\vdots\\y^{(n)} \end{pmatrix}.
\end{align}
Given any function $g:\mathbb{R}\to\mathbb{R}$, we will often use the following notation:
\begin{align}
    g(X)\coloneqq \begin{pmatrix} g(x^{(1)})\\g(x^{(2)})\\\vdots\\g(x^{(n)}) \end{pmatrix},\qquad g(X^T)\coloneqq \begin{pmatrix} g(x^{(1)})&g(x^{(2)})&\cdots& g(x^{(n)}) \end{pmatrix}
\end{align}
Similarly, for any bivariate function $K:\mathbb{R}\times\mathbb{R}\to\mathbb{R}$ we will write $K(X,X^T)$ to indicate the $n\times n$ matrix with entries $\left(K(X,X^T)\right)_{ij}\coloneqq K(x^{(i)},x^{(j)})$ for $1\leq i,j\leq n$. In what follows, we denote by
\begin{equation}
    X = \left(x^{(1)},\,\ldots,\,x^{(n)}\right)^T
\end{equation}
the vector of the training inputs and by $\ov{X}$ the vector of all the possible inputs in $\mathcal{X}$.
We set
\begin{align}\label{fvectmod}
\begin{aligned}
F(t)\coloneqq
\begin{pmatrix}
f(\Theta_{t},x^{(1)})\\
f(\Theta_{t},x^{(2)})\\
\vdots\\
f(\Theta_{t},x^{(n)})
\end{pmatrix} 
\end{aligned}
= f(\Theta_t,X)\,.
\end{align}
From the gradient-flow equation \eqref{gradform1} and the chain rule, the evolution equations for the parameters and the model function can be written as
\begin{align}\label{gradeq1}
 \displaystyle\begin{cases}
  \frac{\de \Theta_{t}}{\de t}=-\eta\nabla_{\Theta}f(\Theta_{t},X^{T})\nabla_{f(\Theta_{t},X)}\L(\Theta_{t}),\\
 \frac{\de}{\de t}f(\Theta_{t},x)=-\eta\left(\nabla_{\Theta}f(\Theta_{t},x)\right)^{T}\nabla_{\Theta}f(\Theta_{t},X^{T})\nabla_{f(\Theta_{t},X)}\L(\Theta_{t})\,,
 \end{cases}   
\end{align}
where $\nabla_{\Theta}f(\Theta_{t},X^{T})$ denotes the gradient of $f(\Theta_{t},X^{T})$ with respect to $\Theta$ while $\nabla_{f(\Theta_{t},X)}\L(\Theta_{t})$ indicates the gradient of the cost function $\L$ with respect to $f(\Theta_{t},X)$. Recall that $^T$ is the transposition operator.
\begin{defn}\label{def:ENTK}
We define the empirical neural tangent kernel, in short  empirical neural tangent kernel, as
\begin{align}\label{NTK1}
\hat{K}_{\Theta}(x,x')\coloneqq \left(\nabla_{\Theta}f(\Theta_{t},x)\right)^{T}\nabla_{\Theta}f(\Theta_{t},x')\,.
\end{align}
\end{defn}
\eqref{gradeq1} can be written as
\begin{align}\label{gradeq2}
\displaystyle\begin{cases}
  \frac{\de \Theta_{t}}{\de t}=-\eta\nabla_{\Theta}f(\Theta_{t},X^{T})\nabla_{f(\Theta_{t},X)}\L(\Theta_{t}),\\
 \frac{\de}{\de t}f(\Theta_{t},X)=-\eta \hat{K}_{\Theta_{t}}(x,X^{T})\nabla_{f(\Theta_{t},X)}\L(\Theta_{t}).
 \end{cases}     
\end{align}

We mention that the quantum neural tangent kernel has been studied in several works inspired by the success of the neural tangent kernels (NTKs) in describing the dynamics of classical neural networks and the convergence of the training processes \cite{shirai2024quantum, incudini2024toward, nakaji2023quantum, liu2022representation, yu2023expressibility, egginger2024hyperparameter}. In \cite{shirai2024quantum}, a quantum NKT for deep quantum circuits is introduced considering a first-order expansion with respect to the parameters since, during training, for a deep enough circuit the parameters do not move much from initialization. In \cite{incudini2024toward}, the authors propose the use of efficient quantum algorithms to construct kernel functions based on properties that are difficult to compute by classical computations. In \cite{nakaji2023quantum}, the authors consider a class of hybrid neural networks composed of a quantum data-encoder followed by a classical network where the quantum part is randomly initialized according to unitary 2-designs, and the classical part is also randomly initialized according to Gaussian distributions, then they apply the NTK theory to construct an effective quantum kernel, which is in general nontrivial to design. In \cite{liu2022representation}, neural tangent kernels for variational circuits are defined, then the authors  derive dynamical equations for their associated loss function in optimization and learning tasks, they analytically solve the dynamics in the lazy training regime. In \cite{yu2023expressibility}, the authors study the connections between the expressibility and value concentration of quantum tangent kernel models. In particular, for global loss functions, they rigorously prove that high expressibility of quantum encodings can lead to exponential concentration of quantum tangent kernel values to zero. In \cite{egginger2024hyperparameter}, there is an analysis on the effects of hyperparameter choice on a quantum kernel model performance and the generalization gap between classical and quantum kernels.

\subsection{The assumptions for the neural tangent kernel}
In what follows, we introduce some useful assumptions to treat the behavior of the neural tangent kernel. Before, let us recall the definition of the analytic neural tangent kernel.
\begin{defn}\label{def:analk1}
We define the analytic neural tangent kernel as the expectation of the empirical neural tangent kernel:
\begin{align}\label{ankd1}
 K(x,x')\coloneqq \E(\hat{K}_{\Theta}(x,x')).  
\end{align}
\end{defn}
\begin{ass}\label{A3}
We suppose that the neural tangent kernel restricted to the training inputs $K=K(X,X^{T})$ is invertible.
We denote with $\lambda_{\max}^K$ and $\lambda_{\min}^K$ its maximum and minimum eigenvalue, respectively.
\end{ass}
\begin{ass}\label{A4}
    Let us furthermore assume that
    \begin{align}\label{eq:A4}
    \lambda_{\min}^{K}&\geq \frac{96 n\sqrt{Lm}|\mathcal{M}|^2|\mathcal{N}|}{N^2(m)}\sqrt{\log 
    \left(2n^2\sqrt{N(m)}\right)},\\
    \frac{Lm\vert \mathcal{M}\vert^{2} \vert\mathcal{N} \vert^{2}}{N^{3}(m)} &\leq \frac{1}{256\log 2}\frac{1}{\ov{N}}. \label{eq:A4bis}
\end{align}
Eq. \eqref{eq:A4} will be needed to enforce that the empirical neural tangent kernel is strictly positive with high probability, while eq. \eqref{eq:A4bis} will be needed to enforce that the empirical neural tangent kernel is close to the analytic neural tangent kernel with high probability.
\end{ass}

\subsection{The expected behavior of $N(m)$ in circuits with logarithmic depth}

In the quantitative bounds that we are going to provide in the next sections, the dependence on the number of qubits will appear as
\begin{align}\label{condizione}
\left(\frac{L^\alpha m^\beta |\mathcal{M}|^\gamma|\mathcal{N}|^\delta}{N(m)}(\log N(m))^\sigma\right)^\nu
\end{align}
for some $\alpha,\beta,\gamma,\delta,\sigma,\nu>0$ according to the precise statement.
As discussed in \cite{girardi2024}, in some cases it is possible to estimate $N(m),|\mathcal{M}|$ and $|\mathcal{N}|$ so that the asymptotical behavior of the bounds for wide circuits can be studied. In the example provided by E. Abedi \& al. in \cite{abedi2023} we have $L=O(1)$ and
\begin{equation}
     N(m)=\sqrt m, \quad |\mathcal{M}|=O(L), \quad  |\mathcal{N}|=O(L^2).
\end{equation}
As we will see, in general $\beta<\frac{1}{2}$, therefore this kind of circuit always satisfies all our theorems. However, since $\dim\,\mathcal{H}_{\mathrm{loc}}=2^{O(L)}=O(1)$, such circuit is classically simulable, so we do not expect to achieve quantum advantages with such architecture.\\
A suitable class of circuits for investigating quantum advantages is the one studied by J. C. Napp in \cite{napp2022quantifying}, where the following bound is provided
\begin{equation}
    N(m)\geq\frac{\sqrt m}{2^{CL}}\quad \text{for some fixed } C. \label{eq:Napp}    
\end{equation} 
Furthermore, we will prove that (see \eqref{eq:inequalityNm})
\begin{equation}
    N(m)\leq \sqrt{m|\mathcal{M}||\mathcal{N}|}.
\end{equation}
In a generic setting the number of qubits in the past light cone of any observable $O_k$ can grow as
\begin{equation}
    |\mathcal{J}_k^1|=O(2^L).
\end{equation}
Under the hypothesis of geometrical locality (i.e., each qubit can interact only with the nearest neighbor qubits), a $d$-dimensional lattice of qubits has
\begin{equation}
    |\mathcal{J}_k^1|=O(L^d).
\end{equation} 
Let us assume to choose $L=\epsilon\log_2 m$. Without assumptions on the geometrical locality, $|\mathcal{J}_k^1|=O(m^{\epsilon})$, which is an upper bound, so any growth $|\mathcal{J}_k^1|=\Theta(m^{\epsilon'})$ with $\epsilon'\leq\epsilon$ can be achieved by an appropriate choice of the interactions. In this case, the dependence on the number of qubits will provide a prefactor \eqref{condizione} asymptotically vanishing provided that $\epsilon,\epsilon'$ are small enough. Indeed, up to multiplicative constants, \eqref{condizione} will behave as
\begin{equation}\label{eq:limit}
    \lim_{m\to\infty}\left(\frac{(\log_2m)^\alpha m^{\beta+(\gamma+\delta)\epsilon'}}{m^{1/2-C\epsilon}}\left(\log m\right)^\sigma\right)^\nu=0
\end{equation} 
since in all our bounds we have $\beta<1/2$. However, the local Hilbert spaces have a quasi-exponential dimension $2^{m^{\epsilon'}}$.
In the geometrically local setting, we can choose $|\mathcal{J}_k^1|=\Theta((\epsilon\log_2m)^{d})$; if $d\geq 2$, then the local Hilbert spaces have a super-polynomial dimension, since $2^{(\epsilon \log_2 m)^d}=m^{ \epsilon^d\log^{d-1}_2 m}$. Besides, \eqref{condizione} is asymptotically vanishing for $\epsilon$ small enough.\\

The fact that there are circuits such that the quantities of the form \eqref{condizione} are asymptotically vanishing -- as in \eqref{eq:limit} -- is fundamental to understand why both conditions of Assumption \autoref{A4} are reasonable. Since in \eqref{eq:A4bis} a quantity of the form \eqref{condizione} appears, we immediately see that it can be satisfied by the circuits described above when the number of qubits is large enough. Similarly, for any fixed $\lambda_0>0$, if we restrict the class of circuits considered above to the ones having $\lambda_{\min}^K\geq \lambda_0$, then \eqref{eq:A4} holds provided that $m$ is large enough.\\

We should mention that the exponential decrease of $N(m)$ on the number of layers is a manifestation of the phenomenon of \textit{barren plateaus} \cite{napp2022quantifying}.
We will not further investigate the particular architectures which satisfy our theorems, but our statements will be motivated by the existence of examples such as the ones we have just presented.

\subsection{The Wasserstein distance of order 1}\label{sub:distances}
In this subsection, we introduce the distance that we employ to quantify the closeness between the probability distribution of the functions generated by quantum neural networks and the associated Gaussian processes (for a precise definition of what a Gaussian process is, see \autoref{sec:gausproc}): the Wasserstein distance of order $1$, also called Monge--Kantorovich distance or earth mover's distance \cite{ambrosio2008gradient,MR2459454,santambrogio2015optimal}.
The Wasserstein distance of order $1$ has been introduced in the context of optimal mass transport  as the minimum cost required to transport a mass distribution into another \cite{monge1781memoire, kantorovich1960mathematical}. It admits applications in several areas of mathematics \cite{villani2009optimal} and in the general scenario of transporting resources in the cheapest way \cite{evans2012phylogenetic, rachev2011probability}. As a distance among probability distributions, the Wasserstein distance finds natural applications in statistics \cite{panaretos2020invitation} and machine learning \cite{frogner2015learning, MAL-073, cheng2020wasserstein, arjovsky2017wasserstein} and it is also be extended to the quantum realm and considered in the context of quantum machine learning \cite{chakrabarti2019quantum,kiani2022}.

Let $(\Omega,\P)$ be a probability space. Let us denote by $\C\coloneqq \C(\R^{d})$ a collection of Borel-measurable real-valued functions on $\R^{d}$. We say that  the class $\C$ is separating if the following property holds: any two $\R^{d}$-valued random variables $F,G$ satisfying $h(F), h(G) \in L^{1}(\Omega,\P)$ and $\E\left(h(F)\right)=\E\left(h(G)\right)$ for any $h\in \C$, have necessarily the same law. The distance between the laws of $F$ and $G$ induced by $\C$ is
\begin{align}\label{gdef1}
\de_{\C}(F,G)\coloneq \sup\left\{\left|\E\left(h(F)\right)- \E\left(h(G)\right)\right|\,: h\in\C\right\}.
\end{align}
The Wasserstein distance of order $1$ between the laws of $F$ and $G$, denoted  by $\de_{\w}(F,G)$, is obtained from \eqref{gdef1} by taking $\C$ to be the set of all functions $h:\R^{d}\rightarrow \R$ such that $\norma{h}_{{\rm Lip}}\leq 1$, where

    \begin{align}\label{dw1}
        \norma{h}_{{\rm Lip}}\coloneqq \sup_{\substack{x\neq y \\ x,y\in\R^{d}}}\frac{\left| h(x)-h(y)\right|}{\norma{x-y}_{2}},
    \end{align}
and $\norma{\cdot}_{2}$ denotes the usual Euclidean norm in $\R^{d}$. The Wasserstein distance of order $1$ admits the following equivalent formulation:
\begin{align}\label{def:wassers1}
\de_{W}(F,G)\coloneqq \inf\left\{ \int_{\R^{d}\times \R^{d}}\norma{x-y}_{2}\de\gamma(x,y): \gamma\in \Gamma({{\rm law}(F)},{{\rm law}(G)})\right\},
\end{align}
where $\Gamma({{\rm law}(F)},{{\rm law}(G)})$ denotes the set of all probability measures supported on $\R^{d}\times \R^{d}$ with marginals ${\rm law}(F)$, and ${\rm law}(G)$. In \eqref{def:wassers1}, we defined the Wasserstein distance of order $1$ between the law of two random variables, $F,G$. In general, this definition applies for any probability measures $\mu,\nu$ on $\R^{d}$. Furthermore, the distance \eqref{def:wassers1} can also be expressed in terms of the so-called Monge--Kantorovich formulation \cite{feldman2002uniqueness}. Let $T:\R^{d}\rightarrow \R^{d}$ be a Borel map, and denote by $T_{\#}\mu$ the push forward of $\mu$ through $T$ defined by 

\begin{align}
    T_{\#}\mu(E)\coloneqq \mu(T^{-1}(E)) \hskip 0,5cm \text{$\forall\, E\subset \R^{d}$, Borel.}
\end{align}
The Monge version of the transport problem is the following: Let $\mu,\nu$ be probability measures on $\R^{d}$ absolutely continuous with respect to the $d$-dimensional Lebesgue measure. Minimize

\begin{align}\label{monge1}
    T\mapsto \int_{\R^{d}\times \R^{d}}\norma{x-T(x)}_{2}{\rm d}\mu(x),
\end{align}
among all the transport maps $T$ from $\mu$ to $\nu$, that is, all maps $T$ such that $T_{\#}\mu=\nu$. It is well-known that \eqref{def:wassers1}  is equivalently written in terms of the minimization problem \eqref{monge1}. That is, there exists an optimal transport map $T:\R^{d}\rightarrow \R^{d}$ such that

\begin{align}
\de_{W}(\mu,\nu)=\int_{\R^{d}}\norma{x-T(x)}_{2}{\rm d}\mu(x).
\end{align}
In what follows, given $s>0$, we will  also need the following truncated version of the Wasserstein distance of order $1$, denoted as $\de_{W}^{(s)}(F,G)$:
\begin{align}\label{def:swassers}
\de_{W}^{(s)}(F,G)\coloneqq \inf\left\{ \int_{\R^{d}\times \R^{d}}\de^{(s)}(x,y)\de\gamma(x,y): \gamma\in \Gamma({{\rm law}(F)},{{\rm law}(G)})\right\},
\end{align}
where 
\begin{align}\label{struncated}
\de^{(s)}(x,y)\coloneq \norma{x-y}_{2}\wedge s, \hskip0,2cm \text{for $x,y\in\R^{d}$}.    
\end{align}
Let us denote by $\C_{{\rm Lip}(1),s}$ the class of $1$-Lipschitz functions with respect to $\de^{(s)}$, and  by $\C_{{\rm Lip}(1)}$ the class of $1$-Lipschitz functions with respect to $\norma{\cdot}_{2}$. Notice that $\C_{{\rm Lip}(1),s}\subset \C_{{\rm Lip}(1)}$, and thus by the dual characterization of the $1$-Wasserstein distance, we get that

\begin{align}\label{ineqtroncata}
\de_{W}^{(s)}(F,G)\leq \de_{W}(F,G).
\end{align}

\section{Some relevant facts about Gaussian Processes}\label{sec:gausproc}
In this part, we collect some of the properties needed in the study of the convergence of a sequence of random variables $\{X_{n}\}_{n\in\N}$ to a Gaussian processes. Let us start by recalling what a stochastic process is.

\begin{defn}
A stochastic process is a collection of random variables ${\left(X_{\alpha}\right)}_{\alpha\in A}$ defined on a same probability space, where $A$ is a set of indices.
\end{defn}
It is customary to write $X\sim Y$ when the random variables $X$, and $Y$ have the same probability distribution. We say that a random variable $X$ is Gaussian, or distributed according to a Gaussian probability distribution $N(\mu,\sigma^{2})$ with $\mu\in \R$, $\sigma\neq 0$, if its probability density is given by
\begin{align*}
p(x)=\frac{1}{\sqrt{2\pi\sigma
^{2}}}e^{-\frac{-(x-\mu)^{2}}{2\sigma^{2}}}, \hskip 0,3cm x\in\R.
\end{align*}
In the vectorial case, given $\mu\in \R^{d}$, and $\K\in \R^{d\times d}$ is positive definite matrix, we write $X\sim N(\mu,\K)$ if $X$ is a random vector with probability density function
\begin{align*}
p(x)=\frac{1}{(2\pi)^{\frac{N}{2}}}\frac{1}{\sqrt{{\rm det}(\K)}}e^{-\frac{1}{2}(x-\mu)^{T}\K^{-1}(x-\mu)}, \hskip 0,3cm x\in\R.
\end{align*}
In what follows, we recall what a Gaussian process is.
\begin{defn}\label{gaussp}
A stochastic process $\{X_{\alpha}\}_{\alpha\in A}$ is called Gaussian process (GP) if for every subset $F\subset A$ there exist a vector $\mu_{F}\in \R^{\vert F\vert}$, and a positive semidefinite covariance matrix $\K_{F}$ such that
\begin{align}\label{defGaussian1}
 \{X_{\alpha}\}_{\alpha\in F}\sim N(\mu_{F},\K_{F}).   
\end{align}
\end{defn}
In particular, if we define the mean function $\mu:A\rightarrow \R$, and the covariance function $\K:A\times A\rightarrow \R$ as
\begin{align}\label{meanv1}
 &\mu(\alpha)\coloneqq\E(X_{\alpha})\\\label{covf1}
 &\K(\alpha,\alpha')\coloneqq \E\left((X_{\alpha}-\mu(\alpha))(X_{\alpha'}-\mu(\alpha'))\right),
\end{align}
then, we can consider $\mu_{F}$, and $\K_{F}$ as the  vector, and the matrix given by the evaluations of these functions, that is, we may set
\begin{align}\label{meanvar1}
    \mu_{F}=\mu(F), \hskip 0,4cm \K_{F}=\K(F,F^{T}).
\end{align}

An important fact about the Gaussian processes is that it can be characterized by the mean and the covariance function. Then, we can refer to a Gaussian process $\{X_{\alpha}\}_{\alpha\in A}$ by writing
\begin{align}\label{gaussp1}
\{X_{\alpha}\}_{\alpha\in A}\sim {\rm GP}(\mu,\K).   
\end{align}
In what follows, we aim to recall some about the convergence in distribution of random variables.

\begin{defn}
We say that a sequence of real random variables $\{X_{n}\}_{n\in\N}$ converges in distribution to a real valued random variable $X$ if for any bounded, and continuous function $f:\R\rightarrow \R$, one has that
\begin{align*}
 \lim_{n\rightarrow +\infty}\E(f(X_{n}))=\E(f(X)).  
\end{align*}
In this case, we denote it as $X_{n}\stackrel{d}{\longrightarrow}X$.
\end{defn}

\section{Multidimensional Stein methods}\label{Sec:Stein}
In this section, we recall some concepts related to Stein's method, which we will employ to study the closeness between the law of a function generated by a quantum circuit and that of a Gaussian random variable. Stein's method is a probabilistic technique that enables the assessment of the distance between two probability measures via second-order differential operators. Specifically, we aim to bound the Wasserstein distance of order one using bounds on solutions to the so-called Stein equations, which we briefly review in this section.\\
In what follows, we set $d\geq 1$, and denote by $\M_{d\times d}$ the set of all real $d\times d$-matrices. Given a positive semidefinite matrix $\Sigma \in \M_{d\times d}$, we denote by $\mathcal{N}_{d}(0,\Sigma)$ the law of a $\R^{d}$-valued Gaussian vector with mean zero and covariance matrix $\Sigma$. When $\Sigma$ is the identity matrix, we write $\I_{d}$.
 In what follows, we closely follows the notation of \cite[Chapter 4]{nourdin2012}. We denote the Hilbert-Schmidt inner product and the Hilbert-Schmidt norm by $\nm{\cdot,\cdot}_{{\rm HS}}$ and $\norma{\cdot}_{{\rm HS}}$, respectively:
 \begin{align}
     \nm{A,B}_{{\rm HS}}\coloneqq \tr(AB^{T}), \hskip 0,5cm\norma{A}_{{\rm HS}}\coloneqq \sqrt{\nm{A,A}_{{\rm HS}}}\qquad\forall\;A,\,B\in\mathbb{M}_{d\times d}\,,
 \end{align}
where $\tr(\cdot)$ denotes the usual trace operator. In the next, we denote by $\norma{A}_{{\rm op}}$ the operator norm of $A$
\begin{align}
  \norma{A}_{{\rm op}}\coloneqq \sup\left\{\text{$\norma{Ax}_{2}: x\in\R^{d}$ such that $\norma{x}_{2}=1$}\right\},
\end{align}
where $\norma{\cdot}_{2}$ is the Euclidean norm of $\R^{d}$. The following lemma, commonly referred to as Stein's Lemma, will be used to derive bounds on the Wasserstein distance of order one between probability measures:
\begin{lem}[{\cite[Lemma 4.1.3]{nourdin2012}}]
Let $\Sigma$ be a positive semidefinite $d\times d$ matrix. Let $N=(N_{1},\ldots, N_{d})$ be a random vector with values in $\R^{d}$. Then $N$ has Gaussian $\mathcal{N}_{d}(0,\Sigma)$ distribution if and only if 

\begin{align}\label{stein1}
 \E\left(\nm{N,\nabla f(N)}_{\R^{d}}\right)=\E\left(\nm{\Sigma,{\rm Hess}\,f(N)}_{{\rm HS}}\right) 
\end{align}
for every $C^{2}$ function $f:\R^{d}\rightarrow\R$ having bounded first and second derivatives. Here, we have denoted by ${\rm Hess}\,f$ the Hessian matrix of $f$.
\end{lem}
In the statement of the previous lemma, there is no need to restrict $\Sigma$ to be positive definite. However, in the next definition, we make such an assumption.
\begin{defn}
Let $\Sigma$ be a positive definite $d\times d$ matrix.
Let $N\sim \mathcal{N}_{d}(0,\Sigma)$. Let 
$h:\R^{d}\rightarrow \R$ be such that $\E\vert h(N)\vert<+\infty$.  The Stein equation associated to $h$ and $N$ is the partial differential equation

\begin{align}\label{pdestein1}
 \nm{\Sigma,{\rm Hess}\,f(x)}_{{\rm HS}}-\nm{x,\nabla\,f(x)}_{\R^{d}}=h(x)-\E\left(h(N)\right).
\end{align}
A solution to the equation \eqref{pdestein1} is a $C^{2}$-function $f$ satisfying \eqref{pdestein1} for every $x\in\R^{d}$.
\end{defn}
The following Proposition relates Lipschitz functions with equation \eqref{pdestein1}.
\begin{prop}[{\cite[Proposition 4.3.2]{nourdin2012}}]\label{boundpec}
Let $N\sim \mathcal{N}_{d}(0,\Sigma)$. Let $h:\R^{d}\rightarrow \R$ be a function with Lipschitz constant $K>0$. Then the function $f_{h}:\R^{d}\rightarrow\R$ given by
\begin{align}\label{solpde1}
 f_{h}(x)\coloneqq \int_{0}^{\infty}\E\left(h(N)-h\left(\exp(-t)x+ \sqrt{1-\exp(-2t)}N\right)\right) \de t
\end{align}
is well-defined, it is a $C^{2}$-function, and satisfies \eqref{pdestein1} for all $x\in\R^{d}$. Furthermore,
\begin{align}\label{boundpde1}
 \sup_{x\in\R^{d}}\norma{{\rm Hess}\,f_{h}(x)}_{{\rm HS}}\leq K\sqrt{d} \norma{\Sigma^{-1}}_{{\rm op}} \norma{\Sigma}_{{\rm op}}^{\frac{1}{2}}.
\end{align}
\end{prop}
The following theorem provides an explicit relationship between the Wasserstein distance of order one and the bound given in \eqref{boundpde1}.

\begin{thm}[{\cite[Theorem 4.4.1]{nourdin2012}}]
Let $\Sigma$ be a positive definite $d\times d$ matrix, and let $N\sim \mathcal{N}_{d}(0,\Sigma)$. For any square integrable random vector $F$ with values in $\R^{d}$ we have
\begin{align}\label{wasest1}
 \de_{\w}(F,N)\leq \sup_{f\in \mathscr{F}_{d}(\Sigma)}\left\vert \E\left(\nm{\Sigma,{\rm Hess}\, f(F)}_{{\rm HS}}\right)- \E\left(\nm{F,\nabla f(F)}_{\R^{d}}\right) \right\vert   
\end{align}
where
\begin{align}\label{wasclass1}
\mathscr{F}_{d}(\Sigma)\coloneqq \left\{\textnormal{$f:\R^{d}\rightarrow \R$ is $C^{2}$, and $\sup_{x\in\R^{d}}\norma{{\rm Hess}\,f(x)}_{{\rm HS}}\leq\sqrt{d}\norma{\Sigma^{-1}}_{{\rm op}} \norma{\Sigma}_{{\rm op}}^{\frac{1}{2}}$}\right\}.  
\end{align}
\end{thm}
This theorem provides an explicit form for bounding the Wasserstein distance of order one. However, this comes at the cost of a less favorable bound with respect to the dimensionality of the ambient space $\R^{d}$,  and the fact that the right-hand side depends on the inverse of the covariance matrix $\Sigma$, which can be challenging to determine explicitly in certain cases \cite{meckes2006}.

\subsection{The analytic solution for the linearized model}\label{sub:linearmod}

By using a linearized version of \eqref{gradeq2}, it is possible to study the behavior of this evolution equation directly. In particular, the first-order approximation of the function $f(\Theta_{t},\cdot)$ enables the analysis of the asymptotic properties of the function itself \cite{abedi2023,girardi2024}. In what follows, we consider
\begin{align}
\displaystyle\begin{cases}
 & \frac{\de \Theta_{t}^{\lin}}{\de t}=-\eta\nabla_{\Theta}f(\Theta_{0},X^{T})\nabla_{f^{\lin}(\Theta_{t},X)}\L^{\lin}(\Theta_{t}^{\lin}),\\
 &\frac{\de}{\de t}f^{\lin}(\Theta_{t}^{\lin},X)=-\eta \hat{K}_{\Theta_{0}}(x,X^{T})\nabla_{f^{\lin}(\Theta_{t}^{\lin},X)}\L^{\lin}(\Theta_{t}^{\lin}),
 \end{cases}     
\end{align}
where the linearized version $f^{\lin}(\Theta_{t}^{\lin},)$ of the system \eqref{gradeq2} is given by

\begin{align}\label{linappr1}
f^{\lin}(\Theta_{t}^{\lin},x)=f(\Theta_{0},x)+\nabla_{\Theta}f(\Theta_{0},x)^{T}(\Theta_{t}^{\lin}-\Theta_{0}).
\end{align}
Furthermore, $\L^{\lin}$ is the loss function $\L$ computed on $\D$ using the linearized model function \eqref{linappr1} instead of the original model function $f(\Theta_{t},\cdot)$. Using the notation \eqref{fvectmod}, we may write the equation for $f(\Theta_{t},\cdot)$, and for $f^{\lin}(\Theta_{t}^{\lin},\cdot)$ as

\begin{align}
&\frac{\de}{\de t}F(t)=-\eta \hat{K}_{\Theta_{t}}(F(t)-Y),\\\label{linevol1}
&\frac{\de}{\de t}F^{\lin}(t)=-\eta \hat{K}_{\Theta_{0}}(F^{\lin}(t)-Y)
\end{align}
where $F^{\lin}(t)\coloneqq f^{\lin}(\Theta_{t}^{\lin},X)$, $\hat{K}_{\Theta_{t}}=\hat{K}_{\Theta_{t}}(X,X^{T})$, for $t\geq 0$. For one moment, suppose that $\hat{K}_{\Theta_{0}}$ is invertible. The solution of \eqref{linevol1} gives the evolution of the linearized function evaluated at the inputs of the dataset. Since $F^{\lin}(0)=F(0)$, we have that

\begin{align}
    F^{\lin}(t)=e^{-\eta t\hat{K}_{\Theta_{0}}}(F(0)-Y)+Y.
\end{align}
For a new input, one gets that 

\begin{align}
\frac{\de }{\de t}f^{\lin}(\Theta_{t}^{\lin},x)=-\eta\hat{K}_{\Theta_{0}}(x,X^{T})e^{-\eta\hat{K}_{\Theta_{0}}t}(F(0)-Y). 
\end{align}
Therefore, 

\begin{align}\label{linevolmod1}
f^{\lin}(\Theta_{t}^{\lin},x)= f(\Theta_{0},x)-\hat{K}_{\Theta_{0}}(x,X^{T})\hat{K}_{\Theta_{0}}^{-1}\left(\mathbbm{1}-e^{-\eta\hat{K}_{\Theta_{0}}t}\right)(F(0)-Y).
\end{align}
Based on this observation, \cite{girardi2024} proved that the probability distribution of $\{f^{\lin}(\Theta_{t},x)\}_{x\in\X}$ converges in distribution to a Gaussian process in the limit of infinite width. This motivates the following notations. We denote by $f^{(\infty)}(\cdot)$ the Gaussian process to which $f(\Theta_{0},\cdot)$ converges in distribution. Let us denote by $\overline{X}$ the vectors of all the inputs of $\mathcal{X}$.
For $t>0$, we set
\begin{align}\label{covlimit}
&\begin{aligned}
\K_{t}(x,x')\coloneqq &\K_{0}(x,x')-K(x,X^{T})K^{-1}\left(\mathbbm{1}-e^{-t\eta\,K}\right)\mathcal{K}_{0}(X,x')\\
&-K(x',X^{T})K^{-1}\left(\mathbbm{1}-e^{-t\eta\,K}\right)\mathcal{K}_{0}(X,x)+\\
&+K(x,X^{T})K^{-1}\left(\mathbbm{1}-e^{-t\eta\,K}\right)\mathcal{K}_{0}(X,X^{T})\left(\mathbbm{1}-e^{-t\eta\,K}\right)K^{-1}K(X,x'),
\end{aligned}\\\label{newmedia}
&\mu_{t}(x)\coloneqq K(x,X^{T})K^{-1}\left(\mathbbm{1}-e^{-t\eta\,K}\right)Y.
\end{align}
More precisely, \cite[Corollary 4.16]{girardi2024} proved that $\{f^{\lin}(\Theta_{t},x)\}_{x\in\X}$ converges to the Gaussian process with mean and covariance defined by \eqref{newmedia}, and \eqref{covlimit}, respectively.

\section{Main results}\label{Sec:results}
In this section we state our main results.
\begin{thm}[Convergence at initialization, formal statement]\label{thmwass1}
Suppose that \autoref{A1}-\autoref{A2} hold true.
Let $\ov{\K}_{0}\coloneqq \K_{0}(\ov{X},\ov{X}^{T})$ be the covariance matrix of the quantum neural network at initialization. Then, we have
\begin{align}\label{ratewasst1}
\de_{\w}\left(f(\Theta,\ov{X}),\mathcal{N}(0,\ov{\K}_{0})\right)\leq\ov{N}^{\frac{3}{2}}\left(2\norma{\ov{\K}_{0}^{-1}}_{{\rm op}}\norma{\ov{\K}_{0}}_{{\rm op}}^{\frac{1}{2}}+6\right)\frac{m|\mathcal{M}|^{7/2}|\mathcal{N}|^{7/2}}{N^3(m)}\left(1+\log N(m)\right),
\end{align}
where $\mathcal{N}(0,\ov{\K}_{0})$ denotes a centered Gaussian process with covariance $\ov{\K}_{0}$. 
\end{thm}
\begin{proof}
    See \autoref{sec:init}.
\end{proof}

\begin{thm}[Convergence of the trained network, formal statement]\label{thmwass4}
Suppose that \autoref{A1}, \autoref{A2}, \autoref{A3}, and \autoref{A4} hold true. Let $\ov{\K}_{0}\coloneqq \K_{0}(\ov{X},\ov{X}^{T})$ be the covariance matrix of the quantum neural network at initialization.
Then, for any $s>0$ and any $t>0$, we have

\begin{align}\label{otherrate4}
 \begin{aligned}
\de_{\w}^{(s)}\left(f(\Theta_{t},\ov{X}),\mathcal{N}(\mu_{t}(\ov{X}),\ov{\K}_{t})\right)&\leq C(\ov{N}, n, \|Y\|_2, \lambda_{\min}^K, \ov{\K}_{0}, s)  \frac{L^2m^{9/4}|\mathcal{M}|^{11/2}|\mathcal{N}|^{9/2}}{N^5(m)}\left(1+\log N(m)\right),
    \end{aligned}
\end{align}
 where 
\begin{align}\label{newcostante}
\begin{aligned}
    C(\ov{N}, n, \|Y\|_2, \lambda_{\min}^K, \ov{\K}_{0}, s)&\coloneqq 96 s\ov{N} +8\sqrt{\ov{N}}\norma{\ov{\K}_{0}^{-1}}_{{\rm op}}\norma{\ov{\K}_{0}}_{{\rm op}}^{\frac{1}{2}}(\ov{N}+4n^2)\\
    &\quad +132n^{2}\sqrt{\ov{N}} \left(2\|Y\|_2^2+4n\right)\,\left(1+\frac{27}{\left(\lambda_{\min}^{K}\right)^{3}}\right)\\
    &\quad+2\left(\sqrt n + \|Y\|_2\right)\left(\frac{1}{\lambda_{\min}^{K}} +\frac{6}{\left(\lambda_{\min}^{K}\right)^{2}}\sqrt{\ov{N} n}\right),
\end{aligned}
\end{align} 
and $\mathcal{N}(\mu_{t}(\ov{X}),\ov{\K}_{t})$ denotes Gaussian process with mean vector $\mu_{t}(\ov{X})$ and covariance matrix $\ov{\K}_{t}\coloneqq \K_{t}(\ov{X},\ov{X}^{T})$, where $\mu_{t}$ and $\K_{t}$  are defined by \eqref{newmedia} and \eqref{covlimit}, respectively.
\end{thm}
\begin{proof}
    See \autoref{sec:train}.
\end{proof}

\begin{rem}
The convergence rates in \autoref{thmwass1} and \autoref{thmwass4} strongly depends on the number of possible inputs $\ov{N}$.
We stress that, if we are interested in only a limited set of inputs, we can improve the rates by restrict $\mathcal{X}$ to such set.
For example, we can choose $\mathcal{X}$ to be the set of the training inputs together with a single test input and get $\ov{N} = n+1$. Furthermore, let us notice that the right hand side of \eqref{otherrate4} does not depend on the training time, and therefore this bound is uniform in time, and holds for finite width even in the limit of $t\rightarrow +\infty$.
\end{rem}

\section{Convergence at initialization}\label{sec:init}
In this Section, we prove \autoref{thmwass1}. 
\begin{proof}
Let $h:\R^{N}\rightarrow \R$ be a Lipschitz function with Lipschitz constant equal to one. Let $\ov{X}$ be the vector of all the inputs in $\mathcal{X}$.
Let
\begin{align}
\begin{aligned}
&W_{i}\coloneqq \sum_{k\notin \mathcal{P}_{i}}\frac{f_{k}(\Theta,\ov{X})}{N(m)}, \hskip 0,2cm Z_{i}\coloneqq \sum_{k\in \mathcal{P}_{i}}\frac{f_{k}(\Theta,\ov{X})}{N(m)}. \hskip 0,2cm \text{for $i=1,\ldots,m$;}\\
&\text{$X_{i}\coloneqq \frac{f_{i}(\Theta,\ov{X})}{N(m)}$ for $i=1,\ldots,m$.}
\end{aligned}
\end{align}
For the sake of a clean notation, let $W\coloneqq f(\Theta,\ov{X})$, and we set $Hf\coloneqq {\rm Hessian}f$ to denote the Hessian of a generic smooth function $f:\R^{d}\rightarrow \R$. Notice that
\begin{align}
 W=W_{i} + Z_{i}.  
\end{align}
Let us denote by $g$ one solution of \eqref{pdestein1} associated to $h$. Further, by \autoref{boundpec} one gets that

\begin{align}\label{steinmod}
\E\left(h\left(W\right)-h(\mathcal{N}(0,\ov{\K}_{0}))\right)=\E\left(\nm{\ov{\K}_{0},Hg\left(W\right)}_{{\rm HS}}-W\cdot \nabla\,g\left(W\right)\right).
\end{align}
Let us now choice a further random variable $\theta$ independent from all the variables $X_{i}$ and uniformly distributed in $[0,1]$. By the fundamental theorem of calculus, we write for $x,u,v\in \R^{\ov{N}}$

\begin{align}
\nm{\nabla\,g(x+v)-\nabla\,g(x),u}=\E\left(\nm{Hg(x+\theta v),uv^{T}}_{\rm HS}\right). 
\end{align}
Since $W_{i}, X_{i}$ are independent, then we obtain

\begin{align}
\begin{aligned}
\E\left(\nm{X_{i},\nabla\, g(W_{i}+Z_{i})}\right)&= \E\left(\nm{X_{i},\nabla\, g(W_{i}+Z_{i})-\nabla\, g(W_{i})}\right)\\
&=\E\left(\nm{Hg(W_{i}+\theta Z_{i}),X_{i}Z_{i}^{T}}_{\rm HS}\right).
\end{aligned}
\end{align}
Then, we insert this equality into \eqref{steinmod} to obtain that

\begin{align}
\E\left(h\left(W\right)-h(\mathcal{N}(0,\ov{\K}_{0}))\right)=\E\left(\nm{\ov{\K}_{0},Hg(W)}_{\rm HS}-\sum_{i=1}^{m}\nm{X_{i}Z_{i}^{T},Hg\left(W_{i}+\theta Z_{i}\right)}_{\rm HS}\right).
\end{align}
Furthermore, we write

\begin{align}\label{tobound1}
\begin{aligned}
\E\left(h\left(W\right)-h(\mathcal{N}(0,\ov{\K}_{0}))\right)&=\E\left(\nm{\ov{\K}_{0}-\sum_{i=1}^{m}X_{i}Z_{i}^{T},Hg(W)}_{\rm HS}\right)\\
&+\E\left(\sum_{i=1}^{m}\nm{X_{i}Z_{i}^{T},Hg(W_{i}+Z_{i})-Hg(W_{i}+\theta Z_{i})}_{\rm HS}\right).
\end{aligned}
\end{align}
Note then that 
\begin{align}
\left\vert\E\sum_{i=1}^{m}\nm{X_{i}Z_{i}^{T},Hg(W_{i}+Z_{i})-Hg(W_{i}+\theta Z_{i})}_{\rm HS}\right\vert&\leq \E\sum_{i=1}^{m}\norma{X_{i}}_{2}\norma{Z_{i}^{T}}_{2}\norma{Hg(W_{i}+ Z_{i})-Hg(W_{i}+\theta Z_{i})}_{\rm op}.
\end{align}
On the other hand, notice that by \autoref{propreg1} with $d=\ov{N}$, we get for $x,y\in \R^{\ov{N}}$,
\begin{align}\label{ourcont}
\norma{Hg(x)-Hg(y)}_{\rm op} \leq  \omega(\ov{N},\norma{x-y}_{2}),
\end{align}
where $\omega$ denotes the modulus of continuity of $Hg$. Hence, by \eqref{ourcont}, and since the maximum number of qubits involved in the definition of $Z_{i}$ are $D$ as defined in \eqref{eq:defZi} and \eqref{maxdegree}, we have that

\begin{align}\label{pass2}
\begin{aligned}
\left\vert\E\sum_{i=1}^{m}\nm{X_{i}Z_{i}^{T},Hg(W_{i}+Z_{i})-Hg(W_{i}+\theta Z_{i})}_{\rm HS}\right\vert&\leq\E\sum_{i=1}^{m}\norma{X_{i}}_{2}\norma{Z_{i}^{T}}_{2} \omega(\ov{N},(1-\theta)\norma{ Z_{i}}_{2})\\
&\leq\E\sum_{i=1}^{m}\norma{X_{i}}_{2}\norma{Z_{i}^{T}}_{2}\omega(\ov{N},\norma{Z_{i}}_{2})\\
&\leq \frac{Dm\ov{N}}{(N(m))^{2}}\omega\left(\ov{N},\frac{D}{N(m)}\right).
\end{aligned}
\end{align}
Let us now proceed to bound the first term on the right hand side of \eqref{tobound1}. Let us define

\begin{align}
C_{i}\coloneqq \E\left(X_{i}Z_{i}^{T}\right),
\end{align}
and observe that $\displaystyle\sum_{i=1}^{m}C_{i}=\E\displaystyle\sum_{i=1}^{m}X_{i}W^{T}=\ov{\K}_{0}$
where we have used that $X_{i}$ and $W_{i}$ are independent, for all $i=1,\ldots, m$. Since $g\in \mathscr{F}_{d}(\ov{\K}_{0})$, we have that $Hg$ satisfies \eqref{boundpde1}, and thus

\begin{align}\label{critical1}
\begin{aligned}
\left\vert\E\nm{\ov{\K}_{0}-\sum_{i=1}^{m}X_{i}Z_{i}^{T},{\rm Hessian}\,g(W)}_{\rm HS}\right\vert &\leq\sqrt{\ov{N}}\norma{\ov{\K}_{0}^{-1}}_{{\rm op}}\norma{\ov{\K}_{0}}_{{\rm op}}^{\frac{1}{2}} \E\norma{\ov{\K}_{0}-\sum_{i=1}^{m}X_{i}Z_{i}^{T}}_{{\rm HS}}\\
&\leq\sqrt{\ov{N}}\norma{\ov{\K}_{0}^{-1}}_{{\rm op}}\norma{\ov{\K}_{0}}_{{\rm op}}^{\frac{1}{2}}\sqrt{\E\norma{\ov{\K}_{0}-\sum_{i=1}^{m}X_{i}Z_{i}^{T}}_{\rm HS}^{2}}\\
&=\sqrt{\ov{N}}\norma{\ov{\K}_{0}^{-1}}_{{\rm op}}\norma{\ov{\K}_{0}}_{{\rm op}}^{\frac{1}{2}}\sqrt{\E\norma{\sum_{i=1}^{m}(C_{i}-X_{i}Z_{i}^{T})}_{{\rm HS}}^{2}}\\
&=\sqrt{\ov{N}}\norma{\ov{\K}_{0}^{-1}}_{{\rm op}}\norma{\ov{\K}_{0}}_{{\rm op}}^{\frac{1}{2}}\sqrt{\E\norma{\sum_{i=1}^{m}M_{i}}_{{\rm HS}}^{2}},
\end{aligned}
\end{align}
where we have defined $M_{i}\coloneqq C_{i}-X_{i}Z_{i}^{T}$, for $i=1,\ldots,m$. Notice that from \autoref{usedremark1}, by the definition of $Z_{i}$, and $Z_{j}$, these are independent whenever $\mathcal{P}_{i}\cap \widetilde{\mathcal{P}}_{j}=\emptyset$. Notice that the same holds true for $M_{i}$ and $M_{j}$ for any $i,j=1,\ldots, m$. Furthermore, since the maximum number of qubits involved in the definition of $Z_{i}$ are $D$ as defined in \eqref{maxdegree}, then
\begin{align}
    \norma{C_{i}}_{\rm HS}\,,\quad \E\norma{X_{i}Z_{i}^{T}}_{\rm HS}\leq \frac{D\ov{N}}{(N(m))^{2}},
\end{align}
so that 
\begin{align}\label{pass3}
    \norma{M_{i}}_{\rm HS}\leq  \frac{2D\ov{N}}{(N(m))^{2}}.
\end{align}
Then, by inserting in \eqref{tobound1},  inequalities \eqref{pass2}, \eqref{critical1}, and \eqref{pass3}, we conclude that
\begin{align}
\begin{aligned}
\E\vert h\left(W\right)-h(\mathcal{N}(0,\ov{\K}_{0}))\vert&\leq \sqrt{\ov{N}}\norma{\ov{\K}_{0}^{-1}}_{{\rm op}}\norma{\ov{\K}_{0}}_{{\rm op}}^{\frac{1}{2}}\sqrt{\E\norma{\sum_{i=1}^{m}M_{i}}_{{\rm HS}}^{2}}+\frac{Dm\ov{N}}{(N(m))^{2}}\omega\left(\ov{N},\frac{D}{N(m)}\right)\\
&\leq \sqrt{\ov{N}}\norma{\ov{\K}_{0}^{-1}}_{{\rm op}}\norma{\ov{\K}_{0}}_{{\rm op}}^{\frac{1}{2}} \sqrt{\frac{\widetilde{D}(2\ov{N}D)^{2}}{(N(m))^{4}}}+\frac{Dm\ov{N}}{(N(m))^{2}}\omega\left(\ov{N},\frac{D}{N(m)}\right)\\\label{modform1}
&\eqqcolon C_{m}^{\ov{N},\ov{\K}_{0}},
\end{aligned}
\end{align}
where in the last inequality, we have used that if $i,j$ are observable such that $M_{i},M_{j}$ are independent, then their components do not give any contribution to the sum on the right hand side of \eqref{modform1} (recall that $\widetilde{D}$ is the constant defined in \eqref{indmeas}, and denotes the maximum over all possible related observables $i,j$). By taking the supremum over all the functions with Lipschitz constant one, we obtain the left-hand side of \eqref{ratewasst1}. That is, 
\begin{align}\label{dnewwass1}
\de_{\w}\left(f(\Theta,\ov{X}),\mathcal{N}(0,\ov{\K}_{0})\right)\leq C_{m}^{\ov{N},\ov{\K}_{0}}.
\end{align}
It remains to bound this constant $C_{m}^{\ov{N},\ov{\K}_{0}}$. First, we notice that
\begin{align}
\nonumber
\mathcal{K}_0(x,x)&=\frac{1}{N^2(m)}\sum_{k,k'=1}^m\mathbb{E}[f_k(\Theta,x)f_{k'}(\Theta,x)]\\
\nonumber&=\frac{1}{N^2(m)}\sum_{k=1}^m\left(\sum_{k'\in\mathcal{P}_k}\mathbb{E}[f_k(\Theta,x)f_{k'}(\Theta,x)]
+\sum_{k'\notin\mathcal{P}_k}\mathbb{E}[f_k(\Theta,x)]\mathbb{E}[f_{k'}(\Theta,x)]\right)\\
\nonumber&=\frac{1}{N^2(m)}\sum_{k=1}^m\sum_{k'\in\mathcal{P}_k}\mathbb{E}[f_k(\Theta,x)f_{k'}(\Theta,x)]\\
\nonumber&\leq \frac{1}{N^2(m)}\sum_{k=1}^m\sum_{k'\in\mathcal{P}_k}\mathbb{E}[|f_k(\Theta,x)||f_{k'}(\Theta,x)|]\\
&\leq \frac{1}{N^2(m)}\sum_{k=1}^m\sum_{k'\in\mathcal{P}_k}1\leq \frac{1}{N^2(m)}m\max_k|\mathcal{P}_k|\leq \frac{1}{N^2(m)}m|\mathcal{M}||\mathcal{N}|.
\end{align}\\
where the last inequality follows from \autoref{stimacone1}. Therefore,
\begin{align}\label{eq:inequalityNm}
    1=\max_{x\in\mathcal{X}}\mathcal{K}_0(x,x)\leq \frac{1}{N^2(m)}m|\mathcal{M}||\mathcal{N}|.
\end{align}
 
 We recall that
\begin{align}
\begin{aligned}
\omega(d,x)\coloneqq 
\begin{cases}
x\left(C(1,d)-2\log\,x\right), & \text{if $ x \leq 1$,}\\
C(1,d), &\text{if $x>1$.}
\end{cases}
\qquad\text{with}\qquad 
 C(1,d)\coloneqq 2^{\frac{3}{2}}\frac{1+2d}{d}\frac{\Gamma\left(\frac{1+d}{2}\right)}{\Gamma\left(\frac{d}{2}\right)}. 
\end{aligned}
\end{align}
%In particular, for any $x>0$,
%\begin{align}
%\begin{aligned}
%\omega(d,x)&\leq 2^{\frac{3}{2}}\frac{1+2d}{d}\frac{\Gamma\left(\frac{1+d}{2}\right)}{\Gamma\left(\frac{d}{2}\right)}+\frac{2}{e}. 
%\end{aligned}
%\end{align}
Recalling Wendel's inequality \cite{Wendel}, valid for any $x>0$ and $s\in (0,1)$,
\begin{align}
    \frac{\Gamma(x+s)}{\Gamma(x)}\leq x^s,
\end{align}
we set $x=d/2$ and $s=1/2$ in order to upper bound
\begin{align}
\begin{aligned}
C(1,d)&\leq 2^{\frac{3}{2}}\frac{1+2d}{d}\sqrt{\frac{d}{2}}\leq 2+4\sqrt{d}. 
\end{aligned}
\end{align}
In particular, for $d\geq 1$ we can upper bound $C(1,d)\leq 6\sqrt{d}$.
Furthermore,
\begin{align}
\begin{aligned}
\omega\left(d,\frac{D}{N(m)}\right)&=  
\begin{cases}
\frac{D}{N(m)}\left(C(1,d)-2\log D+\log N(m)\right) &\text{if $ D \leq N(m)$,}\\
C(1,d) &\text{if $D>N(m)$,}
\end{cases}\\
&\leq
\begin{cases}
\frac{D}{N(m)}\left(C(1,d)+\log N(m)\right) &\text{if $ D \leq N(m)$,}\\
\frac{D}{N(m)}C(1,d) &\text{if $D>N(m)$,}
\end{cases}\\
&\leq \frac{D}{N(m)}C(1,d)\left(1+\log N(m)\right)\\
&\leq 6\sqrt d \frac{D}{N(m)}\left(1+\log N(m)\right),
\end{aligned}
\end{align}
where we have noticed that $d,D,N(m)\geq 1$. This, combined with \autoref{stimacone1} ($D\leq |\mathcal{M}||\mathcal{N}|$) and \autoref{lem:stimatildeD} ($\widetilde D\leq (|\mathcal{M}||\mathcal{N}|)^4$), implies
\begin{align}\label{fcostantona}
\begin{aligned}
    C_{m}^{\ov{N},\ov{\K}_{0}}&\leq  2\ov{N}\sqrt{\ov{N}}\norma{\ov{\K}_{0}^{-1}}_{{\rm op}}\norma{\ov{\K}_{0}}_{{\rm op}}^{\frac{1}{2}}\frac{|\mathcal{M}|^3|\mathcal{N}|^3}{N^2(m)}+6\ov{N}\sqrt{\ov{N}}\frac{m|\mathcal{M}|^2|\mathcal{N}|^2}{N^3(m)}\left(1+\log N(m)\right)\\
    &\leq\ov{N}^{\frac{3}{2}}\left(2\norma{\ov{\K}_{0}^{-1}}_{{\rm op}}\norma{\ov{\K}_{0}}_{{\rm op}}^{\frac{1}{2}}+6\right)\frac{m|\mathcal{M}|^{7/2}|\mathcal{N}|^{7/2}}{N^3(m)}\left(1+\log N(m)\right).
\end{aligned}
\end{align}
where in the second inequality we have used \eqref{eq:inequalityNm}.
\end{proof}

\section{Convergence of the trained network}\label{sec:train}

\subsection{Quantitative lazy training for a finite size circuit}
In this section, we provide a renewed proof of some results in \cite{girardi2024} that takes into account the quantitative contribution of the finite-size effects and where we do not need to build a sequence of networks with diverging width.
Therefore, the results will apply to any network satisfying the expressivity condition \eqref{hp:lambda_pos}.

\begin{thm}[Lazy training, formal statement]
\label{newgradfl}
Suppose that \autoref{A1}, \autoref{A2} and \autoref{A3} hold true, and let
\begin{align}\label{eq:R2}
    R(\delta)&\coloneqq \|Y\|_2+\sqrt{\frac{2n}{\delta}}.
\end{align}
Let us assume that, fixed a constant $0<\delta<1$, it holds that
\begin{align}\label{hp:lambda_pos}
    \lambda_{\min}^{K}\geq \frac{96 n\sqrt{Lm}|\mathcal{M}|^2|\mathcal{N}|}{N^2(m)}\sqrt{\log\frac{2n^2}{\delta}}.
\end{align}
Then, there exist a positive number $\widetilde{\lambda}_{\min}(\delta)$ satisfying
\begin{align}\label{relation1}
    \widetilde{\lambda}_{\min}(\delta)\geq \frac{1}{3}\lambda_{\min}^{K},
\end{align} whose explicit expression is provided in \eqref{eq:deflambda} below, such that, when applying gradient flow with learning rate $\eta$, the following inequalities hold with probability at least $1-\delta$ over random initialization:
\begin{align}
\label{grad1}\mathcal{L} (\Theta_t)&\leq \frac{R^2(\delta)}{2} e^{-2\eta\widetilde{\lambda}_{\min}(\delta)t} &\forall\,  t\geq 0,\\
\label{grad2}\|\Theta_t-\Theta_{0}\|_\infty&\leq\frac{2\sqrt n R(\delta)|\mathcal{M}|}{N(m)}\frac{1}{\widetilde{\lambda}_{\min}(\delta)}\left(1-e^{-\eta\widetilde{\lambda}_{\min}(\delta)t}\right) &\forall\,  t\geq 0,\\
\label{grad3}
\sup_{\substack{x\in\mathcal{X}\\t\geq 0}}|f(\Theta_t,x)-f^{\mathrm{lin}}(\Theta_t^{\mathrm{lin}},x)|&\leq 132 n^2 R^2(\delta)\left(1+\left(\widetilde{\lambda}_{\min}(\delta)\right)^{-3}\right)
\frac{L^2m^2|\mathcal{M}|^5|\mathcal{N}|^2}{N^5(m)}\log N(m).&
\end{align}
\end{thm}

\begin{proof}
We recall Chebyshev's inequality for a random vector $V$ with $\mathbb{E}[V]=v$ and $\mathbb{E}[(X-v)^2]=\sigma^2$:
\begin{align}
    \mathbb{P}\left(\|V-v\|_2\geq k\|\sigma\|_2\right)\leq \frac{1}{k^2}.
\end{align}
Calling $F(0)=f(\Theta_0,X)$, we have $\mathbb{E}[F(0)]=0$ and $\mathbb{E}[F^2(0)]=\text{diag}(\mathcal{K}_0(X,X^T))$. Let us define
\begin{align}
    R(\delta)\coloneqq \|Y\|_2+\sqrt{\frac{2}{\delta}}\left\|\left(\text{diag}(\mathcal{K}_0(X,X^T)\right)^{1/2}\right\|_2
\end{align}
Hence,
\begin{align}
    \nonumber\mathbb{P}\left(\|F(0)-Y\|_2\geq R(\delta)\right)&\leq \mathbb{P}\left(\|F(0)\|_2\geq R(\delta)-\|Y\|_2\right)\\
    \label{eq:corollaryR}&\leq \mathbb{P}\left(\|F(0)\|_2\geq\sqrt{\frac{2}{\delta}}\left\|\left(\text{diag}(\mathcal{K}_0(X,X^T)\right)^{1/2}\right\|_2 \right)\leq \frac{\delta}{2}.
\end{align}
Therefore
\begin{align}
\mathbb{P}\left(\|F(0)-Y\|_2<R(\delta)\right)\geq 1-\frac{\delta}{2}.
\end{align}
Using that 
\begin{align}
    \sum_{i=1}^n\mathcal{K}_0(x^{(i)},x^{(i)})=n,
\end{align}
we have
    \begin{align}\label{R2used}
    R(\delta)\coloneqq \|Y\|_2+\sqrt{\frac{2}{\delta}}\left\|\left(\text{diag}(\mathcal{K}_0(X,X^T)\right)^{1/2}\right\|_2=\|Y\|_2+\sqrt{\frac{2}{\delta}}\left(\sum_{i=1}^n\mathcal{K}_0(x^{(i)},x^{(i)})\right)^{1/2} \leq\|Y\|_2+\sqrt{\frac{2n}{\delta}}.
\end{align}
where we have used that $\mathcal{K}_0(x^{(i)},x^{(i)})\leq 1$.
In particular $R(\delta)\geq 1$, which means $R(\delta)\leq R^2(\delta)$, which will be used later.\\
%Now we provide a finite size version of the proof of \cite[Lemma 3.15]{girardi2024}.
%We recall the following result from \cite{girardi2024}:

For the sake of a lean notation, we define
\begin{align}\label{funzg}
g(\delta)&\coloneqq \frac{16 n\sqrt{Lm}|\mathcal{M}|^2|\mathcal{N}|}{N^2(m)}\sqrt{\log\frac{2n^2}{\delta}},\\
\label{eq:boundrho}
\rho(m)&\coloneqq\frac{1}{2}\frac{N^{2}(m)}{16nLm|\mathcal{M}|^2|\mathcal{N}|}(\lambda_{\min}^{K}-g(\delta))\left(1-\sqrt{1-\frac{128n\sqrt n R(\delta)}{(\lambda_{\min}^{K}-g(\delta))^2}\frac{Lm|\mathcal{M}|^3|\mathcal{N}|}{N^3(m)}}\right) \le \frac{4\sqrt n R(\delta)}{\lambda_{\min}^{K}-g(\delta)}\frac{|\mathcal{M}|}{N(m)},\\
h(\delta)&\coloneqq 16n\frac{Lm}{N^2(m)}\rho(m),\\
\label{eq:deflambda}\widetilde{\lambda}_{\min}(\delta)&\coloneqq \lambda_{\min}^{K}-g(\delta)-h(\delta),
\end{align}
where the bound follows from $1-\sqrt{1-x}\leq x$. 
Let us informally clarify the meaning of these definitions; in particular, the following claims, which will be formally proved later, might help to understand the idea behind the apparently complicated form of \eqref{funzg}-\eqref{eq:deflambda}. We call Spec$(A)$ the spectrum of a symmetric matrix $A$.
\begin{itemize}
    \item The minimal eigenvalue $\lambda_{\min}^{K}$ of the analytic kernel $K(X,X^T)$ is not in general the same one of the empirical kernel at initialization $\hat{K}_{\Theta_0}$ and during the training $\hat{K}_{\Theta_t}$. We will need to lower bound $\min_{0\leq t\leq t_1}\min\text{Spec}(\hat{K}_{\Theta_t})$ and we are going to call $\tilde \lambda_{\min}(\delta)$ our estimate. $t_1$ is a time that will be discussed later. Of course the bound will worsen with the size of the trajectory of $\Theta_t$ ($0\leq t\leq t_1$) in the parameter space: if $\rho(m)$ is the maximal possible radius that such trajectory can achieve with respect to the starting point, $\tilde \lambda_{\min}(\delta)$ will decrease with $\rho(m)$. 
    \item $\tilde \lambda_{\min}(\delta)$ is made of two corrections to $\lambda_{\min}^K$: the first one -- $g(\delta)$ -- emerges from the fact that, even if the empirical neural tangent kernel at initialization is concentrated on the analytic neural tangent kernel, they are different matrices; the second one -- $h(\delta)$ -- is due to the evolution from $\hat K_{\Theta_0}$ to $\hat K_{\Theta_t}$. The latter grows with the radius of the trajectory.
    \item In the proof we will study the trajectory from $t=0$ to $t=t_1$, where the final time is the first instant when the parameter vector $\Theta_{t_1}$ is far from $\Theta_{0}$ more than $\rho(m)$ in the parameter space, where it will be crucial to require $\rho(m)$ to satisfy
    \begin{align}\label{eq:rho}
        \rho(m)=\frac{2\sqrt n R(\delta) |\mathcal{M}|}{N(m)}\frac{1}{\widetilde{\lambda}_{\min}(\delta)}
    \end{align}
    in order to conclude the proof of \eqref{grad1} and \eqref{grad2}.
    Since $\rho(m)$ linearly appears in $\tilde \lambda_{\min}(\delta)$, in order to compute $\rho(m)$ we need to solve the second order equation provided by \eqref{eq:rho}. The solution is \eqref{eq:boundrho}. 
\end{itemize}
Therefore, by our definitions, we have that $\rho(m)$ satisfies \eqref{eq:rho} (which can be explicitly verified by using the definitions provided above).
In particular, by \eqref{eq:boundrho},
\begin{align}
    \widetilde{\lambda}_{\min}(\delta)=\frac{2\sqrt n R(\delta) |\mathcal{M}|}{N(m)}\frac{1}{\rho(m)}\geq \frac{1}{2}(\lambda_{\min}^{K}-g(\delta)),
\end{align}
and, by hypothesis \eqref{hp:lambda_pos}, which can be restated as
\begin{align}
    g(\delta)\leq \frac{1}{6}\lambda_{\min}^{K},
\end{align}
we have that
\begin{align}\label{eq:deflmin}
    \widetilde{\lambda}_{\min}(\delta)\geq \frac{1}{3}\lambda_{\min}^{K}>0.
\end{align}

Let $B_r(\Theta_{0})=\{\Theta: \|\Theta_{0}-\Theta\|_\infty< r\}$ be the ball of center $\Theta_{0}$ and radius $r$. 

Let us now recall the following quantitative result about the closeness between $\hat{K}$ and $K$ proved in \cite{girardi2024}:
\begin{thm}[{\cite[Theorem 4.13]{girardi2024}}]\label{basedthm1}
Suppose that we initialize a parameterized quantum circuit randomly, that is, the parameters $\theta_{i}$ are taken as independent random variables. Then, for any $x,x'\in\X$ we have
\begin{align}\label{tight1}
 \P\left(\vert\hat{K}_{\Theta}(x,x')-K(x,x') \vert\geq \varepsilon\right)\leq\exp\left(-\frac{1}{256}\frac{N^{4}(m)}{Lm\vert \mathcal{M}\vert^{4} \vert\mathcal{N} \vert^{2}}\varepsilon^{2}\right).
\end{align}
\end{thm}
Let $\|\,\cdot\,\|_{{\rm HS}}$ be the Hilbert Schmidt norm and let us apply \autoref{basedthm1} to $F_{ij}\coloneqq K(x^{(i)},x^{(j)})-\hat K_{\Theta_{0}}(x^{(i)},x^{(j)})$, where $1\leq i,j\leq n$, i.e. $F=K-\hat K_{\Theta_{0}}$:
\begin{align}
\nonumber
\mathbb{P}[\|F\|_{{\rm HS}}\geq \epsilon]&=
\mathbb{P}\left[\sum_{i,j=1}^n(F_{ij})^2\geq \epsilon^2\right]\leq \mathbb{P}\left[\max_{ij}|F_{ij}|\geq \epsilon/n\right]\\
&\leq\sum_{i,j=1}^n\mathbb{P}\left[|F_{ij}|\geq \epsilon/n\right]\leq n^2 \exp\left(-\frac{1}{256}\frac{N^{4}(m)}{Lm\vert \mathcal{M}\vert^{4} \vert\mathcal{N} \vert^{2}}\frac{\epsilon^{2}}{n^2}\right)\eqqcolon \frac{\delta}{2},
\end{align}
whence
\begin{align}
    \mathbb{P}\left[\|F\|_{{\rm HS}}< g(\delta)\right]\leq 1-\frac{\delta}{2}.
\end{align}
When $\|F\|_{{\rm HS}}< g(\delta)$, the maximum eigenvalue of $|F|$ is $\lambda_F< g(\delta)$, so
\[ F\preceq |F| \preceq \lambda_F\id\prec g(\delta) \id.\]
The previous equation implies that
\[K-\hat K_{\Theta_{0}} \prec g(\delta) \id \quad \Rightarrow\quad \hat K_{\Theta_{0}} \succ K-g(\delta) \id\succeq (\lambda_{\min}^{K}-g(\delta))\id.\]
So,
\begin{align}
\hat K_{\Theta_{0}}(X,X^T)\succ\left(\lambda_{\min}^{K}-g(\delta)\right)\id 
\end{align}
with probability at least $1-\frac{\delta}{2}$. Let 
\begin{align}
t_1=\inf\left\{t:\|\Theta_t-\Theta_{0}\|_\infty\geq \rho(m)\right\}.
\label{t1}
\end{align}
For $t\leq t_1$ we have, by \cite[Lemma 4.32]{girardi2024}\footnote{since in the current work we do not introduce the normalization $N_K(m)$ for the neural tangent kernel, it is sufficient to set $N_K(m)=1$, and $\Sigma_{1}=Lm$ in the quoted result.}, we have
\begin{align}
|\hat K_{\Theta_t}(x,x')-\hat K_{\Theta_{0}}(x,x')|&\leq 16\frac{Lm|\mathcal{M}|^2|\mathcal{N}|}{N^2(m)}\|\Theta_t-\Theta_{0}\|_\infty\leq16\frac{Lm|\mathcal{M}|^2|\mathcal{N}|}{N^2(m)}\rho(m).
\end{align}
Therefore
\begin{align}\label{ntkbound}
\|\hat K_{\Theta_t}-\hat K_{\Theta_{0}}\|_{{\rm HS}}&\leq 16n\frac{Lm|\mathcal{M}|^2|\mathcal{N}|}{N^2(m)}\rho(m)= h(\delta),
\end{align}
whence
\begin{align}\label{usataforq}
\hat K_{\Theta_t}(X,X^T)\succ\left(\lambda_{\min}^{K}-g(\delta)-h(\delta)\right)\id=\widetilde{\lambda}_{\min}(\delta)\id \qquad \forall \,t\leq t_1,
\end{align}
with probability at least $1-\delta$ (by the union bound applied to the events described above).
 Recalling that
\[
\frac{d}{dt} f(\Theta_t,x)=-\eta\hat K_{\Theta_t}(x,X^T)\cdot \left(F(t)-Y\right),
\]
for $t\leq t_1$ we have, with probability at least $1-\delta$,
\begin{align}
\nonumber \frac{d}{dt}\|F(t)-Y\|_2^2&=-2\eta(F(t)-Y)^T\hat K_{\Theta_t}(F(t)-Y)\leq -2\eta\widetilde{\lambda}_{\min}(\delta)\|F(t)-Y\|_2^2,\\
 \Rightarrow \mathcal{L}(\Theta_t)=\frac{1}{2}\|F(t)-Y\|_2^2&\leq \frac{1}{2}e^{-2\eta\widetilde{\lambda}_{\min}(\delta)t}\|F(0)-Y\|_2^2\leq e^{-2\eta\widetilde{\lambda}_{\min}(\delta)t}\frac{R^2(\delta)}{2}. \label{disug2}
\end{align}
Recalling also that
\[ 
\dot \Theta_t = -\eta \nabla_\Theta f(\Theta_t,X)\cdot (F(t)-Y)
\]
and using \cite[Lemma 4.30]{girardi2024}
\begin{align}
\nonumber \frac{d}{dt}|\theta_i(t)-\theta_i(0)|&\leq \Big|\frac{d}{dt}\theta_i(t)\Big|=\Big|\eta\partial_{\theta_i}f(\Theta,X)\cdot (F(t)-Y)\Big|\\
\nonumber &\leq \eta\|\partial_{\theta_i}f(\Theta,X)\|_2\|F(t)-Y\|_2\\
 &\leq 2 \eta \sqrt n\,\frac{|\mathcal{M}|}{N(m)}\,R(\delta)e^{-\eta\widetilde{\lambda}_{\min}(\delta)t},\\ 
 \Rightarrow |\theta_i(t)-\theta_i(0)|&\leq\frac{2\sqrt n R(\delta)|\mathcal{M}|}{N(m)}\frac{1}{\widetilde{\lambda}_{\min}(\delta)}\left(1-e^{-\eta\widetilde{\lambda}_{\min}(\delta)t}\right)\qquad &\forall\,t\leq t_1,\\
\Rightarrow \|\Theta_t-\Theta\|_\infty&\leq\rho(m)\left(1-e^{-\eta\widetilde{\lambda}_{\min}(\delta)t}\right)\qquad &\forall\,t\leq t_1,
\label{disug3}
\end{align}
with probability at least $1-\delta$, where the last implication follows from \eqref{eq:rho} combined with the definition \eqref{eq:deflmin} of $\widetilde{\lambda}_{\min}(\delta)$. If $t_1<\infty$, then  
\begin{align}
\|\Theta_{t_1}-\Theta\|_\infty&\leq\rho(m)\left(1-e^{-\eta\widetilde{\lambda}_{\min}(\delta)t_1}\right)<\rho(m)\qquad \forall\,t\leq t_1,
\end{align}
with probability at least $1-\delta$, but this contradicts the definition \eqref{t1} of $t_1$, so we must have $t_1=\infty$. This implies that \eqref{disug2} and \eqref{disug3} hold for any $t>0$; so, we have proved \eqref{grad1} and \eqref{grad2}. Combining the latter with \cite[Theorem 4.33]{girardi2024}, we can estimate
\begin{align}
\sup_t|f(\Theta_t,x)-f^{\mathrm{lin}}(\Theta_t,x)|&\leq \frac{Lm |\mathcal{M}|^2|\mathcal{N}|}{N(m)}\|\Theta_t-\Theta_0\|^2_\infty
%\nonumber&\leq 
%\frac{Lm |\mathcal{M}|^2|\mathcal{N}|}{N(m)}
%\left(\frac{R_1}{\lambda_{\min}^{K}}n\sqrt{\log(2n)}\frac{|\mathcal{M}|}{N(m)}\right)^2\\
\leq \frac{Lm |\mathcal{M}|^2|\mathcal{N}|}{N(m)}\rho^2(m),
\end{align}
with probability at least $1-\delta$. Now, we want to prove that
\begin{align}\label{eq:deltaF}
\|F(t)-F^{\mathrm{lin}}(t)\|_2\leq \frac{16nR(\delta)}{\widetilde{\lambda}_{\min}(\delta)} \frac{Lm|\mathcal{M}|^2|\mathcal{N}|}{N^2(m)}\rho(m) \left(1-e^{-\eta\widetilde{\lambda}_{\min}(\delta)t}\right),
\end{align}
with probability at least $1-\delta$. We follow the strategy of \cite{abedi2023}, using our results to improve the final bound. Let us define
\[\Delta(t)=\|F(t)-F^{\mathrm{lin}}(t)\|_2\]
and compute
\begin{align}
    \nonumber
    \frac{1}{2}\frac{d}{dt}\Delta^2(t)%&=\sum_{i=1}^n\frac{1}{2}\frac{d}{dt}\left(f(\Theta_t,x^{(i)})-f^{\mathrm{lin}}(\Theta^{\mathrm{lin}}_t,x^{(i)})\right)^2\\
    %\nonumber
    &= \sum_{i=1}^n  \left(f(\Theta_t,x^{(i)})-f^{\mathrm{lin}}(\Theta^{\mathrm{lin}}_t,x^{(i)})\right)\left(\frac{d}{dt}f(\Theta_t,x^{(i)})-\frac{d}{dt}f^{\mathrm{lin}}(\Theta^{\mathrm{lin}}_t,x^{(i)})\right)\\
    \nonumber
    &= -\eta\sum_{i=1}^n \left(f(\Theta_t,x^{(i)})-f^{\mathrm{lin}}(\Theta^{\mathrm{lin}}_t,x^{(i)})\right) \times\\
    \nonumber
    & \qquad \times \left(\hat K_{\Theta_t}(x^{(i)},X^T)(f(\Theta_t,X)-Y)-\hat K_{\Theta}(x^{(i)},X^T)(f(\Theta^{\mathrm{lin}}_t,X)-Y)\right)\\
    \nonumber
    &= -\eta(F(t)-F^{\mathrm{lin}}(t))^T\hat K_{\Theta_t}(F(t)-Y)+\eta (F(t)-F^{\mathrm{lin}}(t))^T\hat K_{\Theta}(F^{\mathrm{lin}}(t)-Y)\\
    \nonumber
    &= -\eta(F(t)-F^{\mathrm{lin}}(t))^T\hat K_{\Theta_t}(F(t)-Y)\\
    & \quad -\eta (F(t)-F^{\mathrm{lin}}(t))^T\hat K_{\Theta}(F(t)-F^{\mathrm{lin}}(t))+\eta (F(t)-F^{\mathrm{lin}}(t))^T\hat K_{\Theta}(F(t)-Y)
    \label{ultimaeq}
\end{align}
Noticing that $\hat K_{\Theta}$ is positive semidefinite,
\begin{align}
    -\eta (F(t)-F^{\mathrm{lin}}(t))^T\hat K_{\Theta}(F(t)-F^{\mathrm{lin}}(t))\leq 0,
\end{align}
so \eqref{ultimaeq} becomes
\begin{align}
    \nonumber
     \frac{1}{2}\frac{d}{dt}\Delta^2(t)&\leq -\eta(F(t)-F^{\mathrm{lin}}(t))^T\hat K_{\Theta_t}(F(t)-Y)+\eta (F(t)-F^{\mathrm{lin}}(t))^T\hat K_{\Theta}(F(t)-Y)\\
     &=-\eta(F(t)-F^{\mathrm{lin}}(t))^T(\hat K_{\Theta_t}-\hat K_{\Theta})(F(t)-Y)
\end{align}
whence
\begin{align}
    \nonumber
     \left|\Delta(t)\frac{d}{dt}\Delta(t)\right|&\leq \eta\|F(t)-F^{\mathrm{lin}}(t)\|_2\|\hat K_{\Theta_t}-\hat K_{\Theta}\|_{\mathrm{op}}\|F(t)-Y\|_2\\
     &= \eta \Delta(t)\|\hat K_{\Theta_t}-\hat K_{\Theta}\|_{\mathrm{op}}\|F(t)-Y\|_2.
\end{align}
Therefore,
\begin{align} \left|\frac{d}{dt}\Delta(t)\right|&\leq \eta\|\hat K_{\Theta_t}-\hat K_{\Theta}\|_{\mathcal{L}}||F(t)-Y||_2.
\end{align}
By \eqref{disug2}
\begin{align}
    \|F(t)-Y\|_2\leq R(\delta)e^{-\eta\widetilde{\lambda}_{\min}(\delta)t},
\end{align}
and by \eqref{ntkbound}
\begin{align}\label{opdiventaf}
\|\hat K_{\Theta_t}-\hat K_{\Theta_0}\|_{\mathrm{op}}&\leq \|\hat K_{\Theta_t}-\hat K_{\Theta_0}\|_{{\rm HS}}\leq 16n\frac{Lm|\mathcal{M}|^2|\mathcal{N}|}{N^2(m)}\rho(m),
\end{align}
we have
\begin{align}
    |\partial_t\Delta(t)|\leq 16nR(\delta)\eta \frac{Lm|\mathcal{M}|^2|\mathcal{N}|}{N^2(m)}\rho(m)e^{-\eta\widetilde{\lambda}_{\min}(\delta)t}
\end{align}
whence
\[\Delta(t)\leq \frac{16nR(\delta)}{\widetilde{\lambda}_{\min}(\delta)} \frac{Lm|\mathcal{M}|^2|\mathcal{N}|}{N^2(m)}\rho(m) \left(1-e^{-\eta\widetilde{\lambda}_{\min}(\delta)t}\right),\]
and this proves the claim \eqref{eq:deltaF}. Now we have all the ingredients to finish the proof.
We adapt to our case the strategy of \cite{CB18}. Let us compute
\begin{align}
\nonumber
\big\|\dot\Theta_t-\dot\Theta_t^{\mathrm{lin}}\big\|_\infty&=\eta\left\|\nabla_\Theta f(\Theta_t,X^T)(F(t)-Y)-\nabla_\Theta f(\Theta_0,X^T)(F^{\mathrm{lin}}(t)-Y)\right\|_\infty\\
\nonumber
&\leq \eta\left\|\left(\nabla_\Theta f(\Theta_t,X^T)-\nabla_\Theta f(\Theta_0,X^T)\right)(F(t)-Y)\right\|_\infty+\eta\left\|\nabla_\Theta f(\Theta_0,X^T)(F^{\mathrm{lin}}(t)-F(t))\right\|_\infty\\
&\leq \eta\sup_i\|\partial_{\theta_i} f(\Theta_t,X)-\partial_{\theta_i} f(\Theta_0,X)\|_2\|F(t)-Y\|_2+\eta\sup_i\|\partial_{\theta_i} f(\Theta_0,X)\|_2\|F^{\mathrm{lin}}(t)-F(t)\|_2
\end{align}
Let us bound the previous expression term by term. Combining the Lipschitzness result of \cite[Lemma 4.30]{girardi2024} with the lazy training bound \eqref{disug3}, and using the convergence to the examples \eqref{disug2}, we control the first term:
\begin{align}
\eta\sup_i\|\partial_{\theta_i} f(\Theta_t,X)&-\partial_{\theta_i} f(\Theta_0,X)\|_2\|F(t)-Y\|_2\leq 4\eta \sqrt n  R(\delta)\frac{|\mathcal{M}|^2|\mathcal{N}|}{N(m)}\rho(m)e^{-\eta\widetilde{\lambda}_{\min}(\delta)t}=:A(t).
\end{align}
Regarding the second term, we need two different estimates to be used for ``small'' and ``large'' $t$, as we will show soon. The first estimate is based on the Lipschitzness of the gradient \cite[Lemma 4.30]{girardi2024} and on \eqref{eq:deltaF}:
\begin{align}
\eta\sup_i\|\partial_{\theta_i} f(\Theta_0,X)\|_2\|F^{\mathrm{lin}}(t)-F(t)\|_2\leq  \frac{32\eta n\sqrt n R(\delta)}{\widetilde{\lambda}_{\min}(\delta)} \frac{Lm|\mathcal{M}|^3|\mathcal{N}|}{N^3(m)}\rho(m)=:B(t).
\end{align}
The second estimate exploits again \cite[Lemma 4.30]{girardi2024} and the convergence to the examples for the original model \eqref{eq:deltaF}; we also need to quantify the convergence to the examples for the linearized model; this immediately follows from the analytic solution for the evolution of linearized model \eqref{linappr1} and from \eqref{eq:corollaryR}:
\begin{align}
\nonumber
\eta\sup_i\|\partial_{\theta_i} f(\Theta_0,X)\|_2\|F^{\mathrm{lin}}(t)-F(t)\|_2&\leq 2\eta\sqrt n \frac{|\mathcal{M}|}{N(m)}\left(\|F^{\mathrm{lin}}(t)-Y\|_2+\|F(t)-Y\|_2\right)\\
\nonumber
&\leq 2\eta\sqrt n \frac{|\mathcal{M}|}{N(m)}\left(e^{-\eta\lambda_{\min}^{K} t}\|F(0)-Y\|_2+R(\delta) e^{-\eta\widetilde{\lambda}_{\min}(\delta)t}\right)\\
\nonumber
&\leq 2\eta\sqrt n \frac{|\mathcal{M}|}{N(m)}R(\delta)\left(e^{-\eta_0\lambda_{\min}^{K} t}+e^{-\eta\widetilde{\lambda}_{\min}(\delta)t}\right)\\
&\leq 4\eta \sqrt n R(\delta)\frac{|\mathcal{M}|}{N(m)}e^{-\eta\widetilde{\lambda}_{\min}(\delta)t}=:C(t).
\end{align}
Hence
\begin{align}
    \big\|\dot\Theta_t-\dot\Theta_t^{\mathrm{lin}}\big\|_\infty\leq A(t)+B(t) \quad \text{and}\quad \big\|\dot\Theta_t-\dot\Theta_t^{\mathrm{lin}}\big\|_\infty\leq A(t)+C(t).
\end{align}
Defining
\begin{align}
    t^\ast = \frac{1}{\eta \widetilde{\lambda}_{\min}(\delta)}\log N(m),
\end{align}
we integrate
\begin{align}
\nonumber
\big\|\Theta_t-\Theta_t^{\mathrm{lin}}\big\|_\infty&\leq \int_0^\infty A(t)dt + \int_0^{t^\ast} B(t)dt + \int_{t^\ast}^\infty C(t)dt\\
\nonumber
&= 4 \sqrt n  R(\delta)\frac{|\mathcal{M}|^2|\mathcal{N}|}{N(m)}\frac{\rho(m)}{\widetilde{\lambda}_{\min}(\delta)}\\
\nonumber
&\qquad + \frac{32 n\sqrt n R(\delta)}{\widetilde{\lambda}_{\min}(\delta)} \frac{Lm|\mathcal{M}|^3|\mathcal{N}|}{N^3(m)}\frac{\rho(m)}{\widetilde{\lambda}_{\min}(\delta)}\log N(m)\\
\nonumber
&\qquad + \frac{4\sqrt n R(\delta)}{\widetilde{\lambda}_{\min}(\delta)}  \frac{|\mathcal{M}|}{N(m)}e^{-\log N(m)}\\
\label{eq:differenza_cammini}
&=\frac{4\sqrt n R(\delta)}{\widetilde{\lambda}_{\min}(\delta)}\frac{|\mathcal{M}|}{N(m)}\left\{\left(|\mathcal{M}||\mathcal{N}|+\frac{8n}{\widetilde{\lambda}_{\min}(\delta)}\frac{Lm|\mathcal{M}|^2|\mathcal{N}|}{N^2(m)}\log N(m)\right)\rho(m)+\frac{1}{N(m)}\right\}.
\end{align}

So, using \cite[Theorem 4.33]{girardi2024} together with \eqref{disug3} and \eqref{eq:differenza_cammini}, we conclude by computing
\begin{align}
\nonumber
|f(\Theta_t,x)-&f^{\mathrm{lin}}(\Theta_t^{\mathrm{lin}},x)|\\
\nonumber
&=|f(\Theta_t,x)-f(\Theta_0,x)-\nabla_\Theta f(\Theta_0,x)^T(\Theta_t^{\mathrm{lin}}-\Theta_0)|\\
\nonumber
&\leq |f(\Theta_t,x)-f^{\mathrm{lin}}(\Theta_t,x)|+|\nabla_\Theta f(\Theta_0,x)^T(\Theta_t^{\mathrm{lin}}-\Theta_t)|\\
\nonumber
&\leq \frac{Lm|\mathcal{M}|^2|\mathcal{N}|}{N(m)}\|\Theta_t-\Theta_0\|_\infty^2+2\frac{Lm|\mathcal{M}|}{N(m)}\|\Theta_t^{\mathrm{lin}}-\Theta_t\|_\infty\\
\nonumber
&\leq \frac{Lm|\mathcal{M}|^2|\mathcal{N}|}{N(m)}\rho^2(m)\\
\nonumber
&\qquad + \frac{8\sqrt n R(\delta)}{\widetilde{\lambda}_{\min}(\delta)}\frac{Lm|\mathcal{M}|^2}{N^2(m)}\left\{\left(|\mathcal{M}||\mathcal{N}|+\frac{8n}{\widetilde{\lambda}_{\min}(\delta)}\frac{Lm|\mathcal{M}|^2|\mathcal{N}|}{N^2(m)}\log N(m)\right)\rho(m)+\frac{1}{N(m)}\right\}\\
&\leq \frac{20 n R^2(\delta)}{\lambda^2_{\min}(\delta)}\frac{Lm|\mathcal{M}|^4|\mathcal{N}|}{N^3(m)}
+\frac{128 n^2 R^2(\delta)}{\left(\widetilde{\lambda}_{\min}(\delta)\right)^{3}}\frac{L^2m^2|\mathcal{M}|^5|\mathcal{N}|}{N^5(m)}\log N(m)
+ \frac{8\sqrt n R(\delta)}{\widetilde{\lambda}_{\min}(\delta)}\frac{Lm|\mathcal{M}|^2}{N^3(m)}\eqqcolon\xi_m(\delta).
\end{align}
We can use inequality \eqref{eq:inequalityNm} to bound
\begin{align}
\nonumber
\xi_m(\delta)&\leq \frac{n^2 R^2(\delta)}{\widetilde{\lambda}_{\min}(\delta)}\left(8+\frac{20}{\widetilde{\lambda}_{\min}(\delta)}+\frac{128}{\left(\widetilde{\lambda}_{\min}(\delta)\right)^{2}}\right)
\frac{L^2m^2|\mathcal{M}|^5|\mathcal{N}|^2}{N^5(m)}\left(1+\log N(m)\right)\\
\nonumber
&\leq \frac{4n^2 R^2(\delta)}{\widetilde{\lambda}_{\min}(\delta)}\left(2+\frac{5}{\left(\widetilde{\lambda}_{\min}(\delta)\right)^{2}}+\frac{32}{\left(\widetilde{\lambda}_{\min}(\delta)\right)^{3}}\right)
\frac{L^2m^2|\mathcal{M}|^5|\mathcal{N}|^2}{N^5(m)}\left(1+\log N(m)\right)\\
&\leq 132 n^2 R^2(\delta)\left(1+\frac{1}{\left(\widetilde{\lambda}_{\min}(\delta)\right)^{3}}\right)
\frac{L^2m^2|\mathcal{M}|^5|\mathcal{N}|^2}{N^5(m)}\left(1+\log N(m)\right)\,,
\end{align}
where we have recalled that $R(\delta)\leq R^2(\delta)$ and, in the last inequality, we have used that, for $x\geq 0$,
\begin{align}
    2x+5x^2+32x^3\leq 2x+5x^2+32x^3+(x-1)^2+(x-4)^2(x+2) = 33(1+x^3)\,.
\end{align}
This concludes the proof.
\end{proof}

\subsection{The Wasserstein distance between \texorpdfstring{$f$}{f} and \texorpdfstring{$f^{\lin}$}{f\^lin}}
In this part, we prove an upper bound to the Wasserstein distance of order $1$ between the objective function $f(\Theta_{t},x)$ and its linearized version $f^{\lin}(\Theta_{t}^{\lin},x)$.

\begin{lem}\label{f_flin}
Suppose that \autoref{A1}, \autoref{A2}, \autoref{A3} and \autoref{A4} hold true. Let $\ov{X}$ be the vector of all the inputs of $\mathcal{X}$.
Then, for any $s>0$ we have
\begin{align}\label{genwasst2}
\begin{aligned}
\de_{\w}^{(s)}&\left(f(\Theta_{t},\ov{X}),f^{\lin}(\Theta_{t}^{\lin},\ov{X})\right)\\
&\qquad\leq 132n^{2}\sqrt{\ov{N}} \left(\norma{Y}_{2}+\sqrt{2n\sqrt{N(m)}}\right)^{2}\left(1+\frac{27}{\left(\lambda_{\min}^{K}\right)^{3}}\right)
\frac{L^2m^2|\mathcal{M}|^5|\mathcal{N}|^2}{N^5(m)}\left(1+\log N(m)\right)\\
&\qquad\quad +\frac{s}{\sqrt{N(m)}}.
\end{aligned}
\end{align}
\end{lem}
\begin{proof}
Note that
\begin{align*}
\de_{\w}^{(s)}\left(f(\Theta_{t},\ov{X}),f^{\lin}(\Theta_{t}^{\lin},\ov{X})\right)&\leq
\E\min\left\{\norma{f(\Theta_t,\ov{X}) - f^{\lin}(\Theta^\lin_t,\ov{X})}_{2}\,,\;s\right\}\\
&\leq \E\min\left\{\sqrt{\ov{N}}\sup_{x\in\X}\left\vert f(\Theta_{t},x)-f^{\lin}(\Theta_{t}^{\lin},x)\right\vert\,,\;s\right\}\\
\end{align*}
where $\ov{N}$ indicates the number of inputs in $\ov{X}$. Further, let us take
\begin{align}\label{chosendelta}
\delta_{m}\coloneqq \frac{1}{\sqrt{N(m)}};
\end{align}
since \eqref{hp:lambda_pos} holds with $\delta_{m}$, we then apply \autoref{newgradfl}, and so that when applying the gradient flow with learning rate $\eta$, the following inequality holds with probability at least $1-\delta_{m}$ over random initialization:
\begin{align}\label{eq:boundff}
\sup_{\substack{x\in\mathcal{X}\\t\geq 0}}|f(\Theta_t,x)-f^{\mathrm{lin}}(\Theta_t^{\mathrm{lin}},x)|&\leq 132 n^2 R^2(\delta_{m})\left(1+\frac{1}{\left(\widetilde{\lambda}_{\min}(\delta_{m})\right)^{3}}\right)
\frac{L^2m^2|\mathcal{M}|^5|\mathcal{N}|^2}{N^5(m)}\left(1+\log N(m)\right)
\end{align}
where $\widetilde{\lambda}_{\min}(\delta_{m})$ is defined according to \eqref{eq:deflambda}, and satisfies \eqref{relation1} due to \eqref{eq:A4} in \autoref{A4}.
Let us call $\mathcal{E}$ the event in which \eqref{eq:boundff} holds, and let $\mathcal{E}^c$ be its complementary. Clearly $\mathbbm{P}(\mathcal{E})\leq 1$ and $\mathbbm{P}(\mathcal{E}^c)\leq \delta_{m}$, and thus

\begin{align}\label{adjustedf1}
\E\min\left\{\sqrt{\ov{N}}\sup_{x\in\X}\left\vert f(\Theta_{t},x)-f^{\lin}(\Theta_{t}^{\lin},x)\right\vert\,,\;s\right\}
&\leq\E\left(\sqrt{\ov{N}}\sup_{\stackrel{x\in\X}{t\geq 0}}\left\vert f(\Theta_{t},x)-f^{\lin}(\Theta_{t}^{\lin},x)\right\vert \Bigg|\cE\right)+\E\left(s \,|\,\cE^{c}\right)\delta_m.
\end{align}

Hence, by applying \eqref{eq:boundff} and \eqref{adjustedf1}, one gets that

\begin{align}\label{adjustedf2}
\begin{aligned}
\de_{\w}^{(s)}&\left(f(\Theta_{t},\ov{X}),f^{\lin}(\Theta_{t}^{\lin},\ov{X})\right)\\
&\qquad\leq 132\sqrt{\ov{N}} n^2 R^2(\delta_{m})\left(1+\frac{1}{\left(\widetilde{\lambda}_{\min}(\delta_{m})\right)^{3}}\right)
\frac{L^2m^2|\mathcal{M}|^5|\mathcal{N}|^2}{N^5(m)}\left(1+\log N(m)\right)\\
&\qquad\quad +\frac{s}{\sqrt{N(m)}}\\
&\qquad\leq 132\sqrt{\ov{N}} n^2 R^2(\delta_{m})\left(1+\frac{27}{\left(\lambda_{\min}^{K}\right)^{3}}\right)
\frac{L^2m^2|\mathcal{M}|^5|\mathcal{N}|^2}{N^5(m)}\left(1+\log N(m)\right)\\
&\qquad\quad+\frac{s}{\sqrt{N(m)}},
\end{aligned}
\end{align}
where the last inequality follows by \eqref{relation1}. Now by the choice of $\delta_{m}$ and $R(\delta_{m})$ defined according to \eqref{R2used}, we are done.
\end{proof}

\subsection{The closeness between \texorpdfstring{$f^{\lin}$}{f\^lin} and a Gaussian process}
Let us define the following processes:

\begin{align}\label{remaider1}
&R_{t}(\ov{X})\coloneqq -\hat{K}_{\Theta_{0}}(\ov{X},X^{T})\hat{K}_{\Theta_{0}}^{-1}\left(\mathbbm{1}-{\rm e}^{-\eta\,t\hat{K}_{\Theta_{0}}}\right)(F(0)-Y),\\\label{limremaider}
&R_{t}^{\infty}(\ov{X})\coloneqq -K(\ov{X},X^{T})K^{-1}\left(\mathbbm{1}-{\rm e}^{-\eta\,tK}\right)(f^{(\infty)}(X)-Y).
\end{align}
In the next, we show that for any $t>0$, $f^{\lin}(\Theta_{t},\ov{X})$ is close to $Z_{t}(\ov{X})$ that is defined as
\begin{align}
\label{olimite1}
  &Z_{t}(\ov{X})\coloneqq f^{(\infty)}(\ov{X})-K(\ov{X},X^{T})K^{-1}\left(\mathbbm{1}-{\rm e}^{-\eta\,tK}\right)(f^{(\infty)}(X)-Y).
\end{align}
By $Z_{t}(\ov{X})$, we denote a stochastic process that follows a normal probability law $\mathcal{N}(\mu_{t}(\ov{X}),\K_{t}(\ov{X},\ov{X}^{T}))$, where $\mu_{t}(\ov{X})$ represents its mean vector, and $\ov{\K}_{t}\coloneqq \K_{t}(\ov{X},\ov{X}^{T})$  its covariance operator. These quantities are the vectorial form of \eqref{newmedia}, and \eqref{covlimit}, respectively. 
\begin{lem}\label{thmwass2}
Let us fix $0<s<+\infty$. Let $\ov{X}$ be the vector of all the inputs in $\mathcal{X}$, and let $X$ be the vector of the $n$ training inputs. For any $t>0$, we set $\ov{\K}_{t}\coloneqq\K_{t}(\ov{X},\ov{X}^{T})$ defined according to \eqref{covlimit}. Suppose that \autoref{A1}, \autoref{A2} and \autoref{A3} hold true. Then 

\begin{align}\label{ratewasst2}
\begin{aligned}
\de_{\w}^{(s)}\left(f^{\lin}(\Theta_{t},\ov{X}),\mathcal{N}(\mu_{t}(\ov{X}),\ov{\K}_{t})\right) &\leq  C_{m}^{\ov{N},\ov{\K}_{0}}+\frac{\norma{K(\ov{X},X^{T})}_{{\rm op}}}{\lambda_{\min}^{K}}C_{m}^{n,\ov{\K}_{0}}+\\
&+2\left(\frac{1}{\sqrt{N(m)}\lambda_{\min}^{K}} +\frac{6}{\sqrt{N(m)}\left(\lambda_{\min}^{K}\right)^{2}}\sqrt{\bar N n}\frac{Lm|\mathcal{M}|^2}{N^2(m)}\right)\E\norma{F(0)-Y}_{2}\\
&\phantom{formula}+\exp\left(-\frac{1}{256}\frac{N^{3}(m)}{Lm\vert \mathcal{M}\vert^{2} \vert\mathcal{N} \vert^{2}}\frac{1}{\ov{N}}\right)s
\end{aligned}
\end{align}
where $C_{m}^{\ov{N},\ov{\K}_{0}}$ is defined according to \eqref{modform1}, and $C_{m}^{n,\ov{\K}_{0}}$ defined as

\begin{align}\label{modform2}
C_{m}^{n,\ov{\K}_{0}}\coloneqq  \sqrt{n}\norma{\ov{\K}_{0}^{-1}}_{{\rm op}}\norma{\ov{\K}_{0}}_{{\rm op}}^{\frac{1}{2}} \sqrt{\frac{\widetilde{D}(2nD)^{2}}{(N(m))^{4}}}+\frac{Dmn}{(N(m))^{2}}\omega\left(n,\frac{D}{N(m)}\right),
\end{align}
where $D$, $\widetilde{D}$, and $\omega$ are defined in \eqref{maxdegree}, \eqref{indmeas}, and \eqref{contfunct}, respectively.
\end{lem}
\begin{proof}
Consider $0<s<+\infty$, and let $h$ be a 1-Lipschitz function. Let us define

\begin{align}
\begin{aligned}\label{auxlinearf}
&\tilde{f}^{\lin}(\Theta_{t}^{\lin},\ov{X})=f(\Theta_{0},\ov{X})+ \tilde{R}_{t}(\ov{X}),\\
&\tilde{R}_{t}(\ov{X})\coloneqq -K(\ov{X},X^{T})K^{-1}\left(\mathbbm{1}-{\rm e}^{-\eta\,tK}\right)(F(0)-Y).
\end{aligned}
\end{align}
Notice that by the triangle inequality

\begin{align}
\de_{W}^{(s)}(f^{\lin}(\Theta_{t}^{\lin},\ov{X}),Z_{t})\leq  \de_{W}^{(s)}(f^{\lin}(\Theta_{t}^{\lin},\ov{X}),\tilde{f}^{\lin}(\Theta_{t}^{\lin},\ov{X})) + \de_{W}^{(s)}(\tilde{f}^{\lin}(\Theta_{t}^{\lin},\ov{X}),Z_{t})
\end{align}
and so
\begin{align}\label{dalimt1}
\begin{aligned}
\de_{W}^{(s)}(f^{\lin}(\Theta_{t}^{\lin},\ov{X}),Z_{t})&\leq \E\min\left\{\norma{f^{\lin}(\Theta_{t}^{\lin},\ov{X})-\tilde{f}^{\lin}(\Theta_{t}^{\lin},\ov{X})}_{2},s\right\} +\\
&+\E\min\left\{\norma{\tilde{f}^{\lin}(\Theta_{t}^{\lin},\ov{X})-Z_{t}}_{2},s\right\}.
\end{aligned}
\end{align}
Notice that by the definition of  $f$ and $\tilde{f}^{\lin}$, we have

\begin{align}\label{primobound}
 \E\min\left\{\norma{f^{\lin}(\Theta_{t}^{\lin},\ov{X})-\tilde{f}^{\lin}(\Theta_{t}^{\lin},\ov{X})}_{2},s\right\}\leq \E\min\left\{\norma{R_{t}-\tilde{R}_{t}}_{2},s\right\}.  
\end{align}
In what follows, since $\eta\,$ is a constant, we can absorb it into the variable $t$, and therefore, we will not consider it further. Hence, for each matrix $\Sigma$ depending on $\ov{X},X$ such that  the inverse matrix $(\Sigma(X,X^{T}))^{-1}$ is well defined, we define the matricial function,

\begin{align}
M(\Sigma)=\Sigma(\ov{X},X^{T})\Sigma^{-1}(\mathbbm{1}-e^{-t\Sigma}), \hskip 0,1cm \text{where $\Sigma=\Sigma(X,X^{T})$, and $\Sigma^{-1}=(\Sigma(X,X^{T}))^{-1}$.}
\end{align}
For the sake of simplicity, given a covariance matrix $\Sigma$ of a set of real valued random variables indexed by the elements of $\mathcal{X}$,  we write $\Sigma_{\ast}\coloneqq  \Sigma(\ov{X},X^{T})$. Then
\begin{align}
M(K)=K_{\ast}K^{-1}(\mathbbm{1}-e^{-tK}), \hskip 0,2cm M(\hat{K})=\hat{K}_{\ast}\hat{K}^{-1}(\mathbbm{1}-e^{-t\hat{K}}), \hskip 0,1cm \text{where $\hat{K}=K_{\Theta_{0}}$}.
\end{align}

Notice that 

\begin{align}
\de M=\de \Sigma_{\ast}\Sigma^{-1}(\mathbbm{1}-e^{-t\Sigma})-\Sigma_{\ast}\Sigma^{-1}\de \Sigma \Sigma^{-1}(1-e^{-t\Sigma}) + \Sigma_{\ast}\Sigma^{-1}\int_{0}^{1}e^{-s\Sigma t}t\de \Sigma e^{-(1-s)\Sigma t}\de s
\end{align}
Let us denote by $\lambda_{\min}^{\Sigma}$ the smallest eigenvalue of $\Sigma$, and suppose that $\Sigma \geq \lambda_{\min}^{\Sigma}\mathbbm{1}$. Notice that

\begin{align}\label{f:disug1}
\begin{aligned}
\norma{\de M}_{{\rm op}}&\leq \norma{\de \Sigma_{\ast}}_{{\rm op}}\frac{1}{\lambda_{\min}^{\Sigma}} +\norma{\Sigma_{\ast}}_{{\rm op}}\norma{\de \Sigma}_{{\rm op}}\frac{1}{(\lambda_{\min}^{\Sigma})^{2}}+t\norma{\Sigma_{\ast}}_{{\rm op}}\frac{1}{\lambda_{\min}^{\Sigma}}\norma{\de \Sigma}_{{\rm op}}e^{-t\lambda_{\min}^{\Sigma}} \\ 
&\leq \norma{\de \Sigma_{\ast}}_{{\rm op}}\frac{1}{\lambda_{\min}^{\Sigma}} +\norma{\Sigma_{\ast}}_{{\rm op}}\norma{\de \Sigma}_{{\rm op}}\frac{1}{(\lambda_{\min}^{\Sigma})^{2}}+ \frac{1}{(\lambda_{\min}^{\Sigma})^{2}e}\norma{\Sigma_{\ast}}_{{\rm op}}\norma{\de \Sigma}_{{\rm op}}.
\end{aligned}
\end{align}
On the other hand, by the definition of $R_{t}$, and $\tilde{R}_{t}$
\begin{align}
\E\min\left\{\norma{R_{t}-\tilde{R}_{t}}_{2},s\right\}=\E\min\left\{\norma{\left(M(\hat{K})-M(K)\right)(F(0)-Y)}_{2},s\right\},
\end{align}
where $\hat{K}=\hat{K}_{\Theta_{0}}$ is the empirical tangent kernel, and $K$ the analytic tangent kernel. In what follows, fix $0<\varepsilon$, and we call $\cE\coloneqq \left\{\omega\in\Omega: \norma{K( \ov{X},\ov{X}^{T})-\hat{K}(\ov{X},\ov{X}^{T})}_{{\rm op}} <\varepsilon\right\}$. We write 

\begin{align}
\begin{aligned}
\E\min\left\{\norma{\left(M(\hat{K})-M(K)\right)(F(0)-Y)}_{2},s\right\}=& \P\left(\cE\right)\E\left(\min\left\{\norma{\left(M(\hat{K})-M(K)\right)(F(0)-Y)}_{2},s\right\}\Bigg| \cE\right)+\\
+&\P\left(\cE^{c}\right)\E\left(\min\left\{\norma{\left(M(\hat{K})-M(K)\right)(F(0)-Y)}_{2},s\right\}\Bigg| \cE^{c}\right),
\end{aligned}
\end{align}
and thus

\begin{align}
\begin{aligned}
\E\min\left\{\norma{R_{t}-\tilde{R}_{t}}_{2},s\right\}\leq & \P\left(\cE\right)\E\left(\min\left\{\norma{\left(M(\hat{K})-M(K)\right)(F(0)-Y)}_{2},s\right\}\Bigg| \cE\right)+ \P\left(\cE^{c}\right)s\\
&\leq \P\left(\cE\right)\E\left(\min\left\{\norma{\left(M(\hat{K})-M(K)\right)}_{{\rm op}}\norma{(F(0)-Y)}_{2},s\right\}\Bigg| \cE\right)+ \P\left(\cE^{c}\right)s.
\end{aligned}
\end{align}
Let us take $\alpha\in [0,1]$. By the fundamental theorem of calculus, we write

\begin{align}\label{f:disug2}
 M(K)-M(\hat{K})=\int_{0}^{1}\frac{dM(\Sigma_{\alpha})}{\de \alpha}\de \alpha, \hskip 0,2cm \Sigma_{\alpha}\coloneqq \alpha K+(1-\alpha)\hat{K}.   
\end{align}
By \eqref{f:disug1} we have that

\begin{align}
\begin{aligned}    
\norma{M(K)-M(\hat{K})}_{{\rm op}}&\leq \norma{\frac{\de (\Sigma_{\alpha})_{\ast}}{\de \alpha}}_{{\rm op}}\frac{1}{\lambda_{\min}^{K}} +\norma{K_{\ast}}_{{\rm op}}\norma{\frac{\de \Sigma_{\alpha}}{\de \alpha}}_{{\rm op}}\frac{1}{(\lambda_{\min}^{K})^{2}}+ \frac{1}{(\lambda_{\min}^{K})^{2}e}\norma{K_{\ast}}_{{\rm op}}\norma{\frac{\de \Sigma_{\alpha}}{\de \alpha}}_{{\rm op}}\\
&=\frac{1}{\lambda_{\min}^{K}}\norma{K(\ov{X},X^{T})-K_{\Theta_{0}}(\ov{X},X^{T})}_{{\rm op}}+\frac{1}{\left(\lambda_{\min}^{K}\right)^{2}}\norma{K_{\ast}}\norma{K-\hat{K}}_{{\rm op}}+\\
&+ \frac{1}{(\lambda_{\min}^{K})^{2}e}\norma{K_{\ast}}_{{\rm op}}\norma{K-\hat{K}}_{{\rm op}}.
\end{aligned}
\end{align}
Hence, we obtain that
\begin{align}
&\E\left(\min\left\{\norma{M(\hat{K})-M(K)}_{{\rm op}}\norma{F(0)-Y}_{2},s\right\}\Bigg| \cE\right)\leq \frac{\varepsilon}{\lambda_{\min}^{K}}\E\left(\min\left\{\norma{F(0)-Y}_{2},s\right\}\Bigg| \cE\right)+\\
&+\frac{\varepsilon}{\left(\lambda_{\min}^{K}\right)^{2}}\E\left(\min\left\{\norma{K_{\ast}}_{{\rm op}}\norma{F(0)-Y}_{2},s\right\}\Bigg| \cE\right)+ \frac{\varepsilon}{(\lambda_{\min}^{K})^{2}e}\E\left(\min\left\{\norma{K_{\ast}}_{{\rm op}}\norma{F(0)-Y}_{2},s \right\}\Bigg| \cE\right).
\end{align}
In order to bound $\norma{K_{\ast}}_{{\rm op}}$ and $\norma{\hat K_{\ast}}_{{\rm op}}$, let us first notice that, for a generic real $n_1\times n_2$ matrix $A$ such that $|A_{ij}|\leq a$ for any $1\leq i\leq n_1$ and $1\leq j\leq n_2$, by convexity of $x\mapsto x^2$ we have
\begin{align}\label{eq:opHS}
    \|A\|_{\rm op} = \sup_{\|v\|\leq 1}\sqrt{\sum_{i=1}^{n_1}\left(\sum_{j=1}^{n_2}A_{ij}v_j\right)^2}\leq \sup_{\|v\|\leq 1}\sqrt{n_2\sum_{i=1}^{n_1}\sum_{j=1}^{n_2}A_{ij}^2v_j^2}
    \leq a \sup_{\|v\|\leq 1}\sqrt{n_1n_2\sum_{j=1}^{n_2}v_j^2}
    \leq a\sqrt{n_1n_2}
\end{align}
where $v\in \mathbb{R}^{n_2}$. A uniform bound on the matrix elements of $\hat{K}_{\D}$ can be obtained as follows. By \cite[Lemma 4.30]{girardi2024}, for any $1\leq i\leq Lm$ and $x\in\mathcal{X}$ we have
\begin{align}
    |\partial_{\theta_i}f(\Theta_0,x)|\leq 2\frac{|\mathcal{M}|}{N(m)},
\end{align}
therefore
\begin{align}
    |\hat{K}_{\Theta_0}(x,x')|=\left|\sum_{j=1}^{Lm}\partial_{\theta_j}f(\Theta_0,x)\partial_{\theta_j}f(\Theta_0,x')\right|\leq \sum_{j=1}^{Lm}|\partial_{\theta_j}f(\Theta_0,x)|\,|\partial_{\theta_j}f(\Theta_0,x')|\leq 4\frac{Lm|\mathcal{M}|^2}{N^2(m)}
\end{align}
whence
\begin{align}\label{explitcb}
    \|\hat{K}_{\ast}\|_{\rm op}\leq 4\sqrt{\bar N n}\frac{Lm|\mathcal{M}|^2}{N^2(m)}\qquad \text{and}\qquad \|K_{\ast}\|_{\rm op}\leq \mathbb{E}[\|\hat K_{\D}\|_{\rm op}]\leq 4 \sqrt{\bar N n}\frac{Lm|\mathcal{M}|^2}{N^2(m)}.
\end{align}

Hence, we obtain that
\begin{align}\label{formin1}
\begin{aligned}
&\E\left(\min\left\{\norma{M(\hat{K})-M(K)}_{{\rm op}}\norma{F(0)-Y}_{2},s\right\}\Bigg| \cE\right)\leq \frac{\varepsilon}{\lambda_{\min}^{K}}\E\left(\min\left\{\norma{F(0)-Y}_{2},s\right\}\Bigg| \cE\right)\\
&\phantom{formula}+\frac{3}{2}\frac{\varepsilon}{\left(\lambda_{\min}^{K}\right)^{2}}\E\left(\min\left\{4 \sqrt{\bar N n}\frac{Lm|\mathcal{M}|^2}{N^2(m)}\norma{F(0)-Y}_{2},s\right\}\Bigg| \cE\right)\\
&\le \left(\frac{\varepsilon}{\lambda_{\min}^{K}} +\frac{6\varepsilon}{\left(\lambda_{\min}^{K}\right)^{2}}\sqrt{\bar N n}\frac{Lm|\mathcal{M}|^2}{N^2(m)}\right)\E\left(\norma{F(0)-Y}_{2} | \cE\right).
\end{aligned}
\end{align}
Let us now compute $\P(\cE^{c})$.
Proceeding as in \eqref{eq:opHS} we get
\begin{align}\label{upperbd1}
\begin{aligned}
 \left\|K(\ov{X},\ov{X}^{T})-\hat{K}(\ov{X},\ov{X}^{T})\right\|_{\rm op} \leq \sqrt{\ov{N}}\sup_{x,x'\in \X}\vert K(x,x')-\hat{K}(x,x')\vert.
\end{aligned}
\end{align}
Then, by \autoref{basedthm1} and \eqref{upperbd1}, one gets
\begin{align}\label{formin2}
\begin{aligned}
\P(\cE^{c})&\leq \P\left(\sqrt{\ov{N}}\sup_{x,x'\in \X}\vert K(x,x')-\hat{K}(x,x')\vert \geq\varepsilon\right)\\
&=\P\left(\sup_{x,x'\in \X}\vert K(x,x')-\hat{K}(x,x')\vert \geq\frac{\varepsilon}{\sqrt{\ov{N}}}\right)\\
&\leq \exp\left(-\frac{1}{256}\frac{N^{4}(m)}{Lm\vert \mathcal{M}\vert^{2} \vert\mathcal{N} \vert^{2}}\frac{\varepsilon^{2}}{\ov{N}}\right).
\end{aligned}
\end{align}

Therefore, by \eqref{primobound} together with \eqref{formin1} and \eqref{formin2}, one has 
\begin{align}\label{eq:bounddop2}
\begin{aligned}
&\E\min\left\{\norma{f^{\lin}(\Theta_{t}^{\lin},\ov{X})-\tilde{f}^{\lin}(\Theta_{t}^{\lin},\ov{X})}_{2},s\right\}\leq 
\left(\frac{\varepsilon}{\lambda_{\min}^{K}} +\frac{6\varepsilon}{\left(\lambda_{\min}^{K}\right)^{2}}\sqrt{\bar N n}\frac{Lm|\mathcal{M}|^2}{N^2(m)}\right)\E\left(\norma{F(0)-Y}_{2} | \cE\right)\\
&\phantom{formula}+\exp\left(-\frac{1}{256}\frac{N^{4}(m)}{Lm\vert \mathcal{M}\vert^{2} \vert\mathcal{N} \vert^{2}}\frac{\varepsilon^{2}}{\ov{N}}\right)s
\end{aligned}
\end{align}
Let us now estimate the remaining term $\E\min\left\{\norma{\tilde{f}^{\lin}(\Theta_{t}^{\lin},\ov{X})-Z_{t}}_{2},s\right\}$. Notice that
\begin{align}
\tilde{f}^{\lin}(\Theta_{t}^{\lin},\ov{X}))- Z_{t}= f(\Theta_{0},\ov{X})-f^{(\infty)}(\ov{X}) -\left(K(\ov{X},X^{T})K^{-1}\left(\mathbbm{1}-{\rm e}^{-\eta\,tK}\right)(f(\Theta_{0},X)-f^{(\infty)}(X))\right),
\end{align}
and thus
\begin{align}
\begin{aligned}
\E\min\left\{\norma{\tilde{f}^{\lin}(\Theta_{t}^{\lin},\ov{X})-Z_{t}}_{2},s\right\}\leq \E\norma{f(\Theta_{0},\ov{X})-f^{(\infty)}(\ov{X})}_{2} + \frac{\norma{K(\ov{X},X^{T})}_{{\rm op}}}{\lambda_{\min}^{K}}\E\norma{f(\Theta_{0},X)-f^{(\infty)}(X)}_{2}\\
= \de_{W}(f(\Theta_{0},\ov{X}), \mathcal{N}(0,\ov{\K}_{0}))+ \frac{\norma{K(\ov{X},X^{T})}_{{\rm op}}}{\lambda_{\min}^{K}}\de_{W}(f(\Theta_{0},X), \mathcal{N}(0,\ov{\K}_{0}))
\end{aligned}
\end{align}
where the last identity follows by the existence of an optimal plan. Then by \eqref{dnewwass1}, one gets that
\begin{align}\label{quasifit}
\E\min\left\{\norma{\tilde{f}^{\lin}(\Theta_{t}^{\lin},\ov{X})-Z_{t}}_{2},s\right\}\leq C_{m}^{\ov{N},\ov{\K}_{0}}+ \frac{\norma{K(\ov{X},X^{T})}_{{\rm op}}}{\lambda_{\min}^{K}}C_{m}^{n,\ov{\K}_{0}}.
\end{align}
Therefore, combining \eqref{eq:bounddop2} and \eqref{quasifit} with \eqref{dalimt1}, we obtain

\begin{align}
\begin{aligned}
\de_{\w}^{(s)}\left(f^{\lin}(\Theta_{t},\ov{X}),\mathcal{N}(\mu_{t}(\ov{X}),\ov{\K}_{t})\right) &\leq  C_{m}^{\ov{N},\ov{\K}_{0}}+\frac{\norma{K(\ov{X},X^{T})}_{{\rm op}}}{\lambda_{\min}^{K}}C_{m}^{n,\ov{\K}_{0}}+\\
&+\left(\frac{\varepsilon}{\lambda_{\min}^{K}} +\frac{6\varepsilon}{\left(\lambda_{\min}^{K}\right)^{2}}\sqrt{\bar N n}\frac{Lm|\mathcal{M}|^2}{N^2(m)}\right)\E\left(\norma{F(0)-Y}_{2} | \cE\right)\\
&\phantom{formula}+\exp\left(-\frac{1}{256}\frac{N^{4}(m)}{Lm\vert \mathcal{M}\vert^{2} \vert\mathcal{N} \vert^{2}}\frac{\varepsilon^{2}}{\ov{N}}\right)s.
\end{aligned}
\end{align}
Let us choose $\varepsilon=\frac{1}{\sqrt{N(m)}}$. Notice that 
\begin{align}
\E\left(\norma{F(0)-Y}_{2} | \cE\right)\leq \frac{1}{\mathbb{P}(\cE)}\E\norma{F(0)-Y}_{2} \le 2\,\E\norma{F(0)-Y}_{2}\,,   
\end{align}
where
\begin{align}
    \mathbb{P}(\cE)=1-\mathbb{P}(\cE^c)\geq 1-\exp\left(-\frac{1}{256}\frac{N^{3}(m)}{Lm\vert \mathcal{M}\vert^{2} \vert\mathcal{N} \vert^{2}}\frac{1}{\ov{N}}\right)\geq \frac{1}{2}
\end{align}
which follows from \autoref{A4}, and our conclusion follows.
\end{proof}
We now are in position to prove \autoref{thmwass4}.
\begin{proof}[{Proof of \autoref{thmwass4}}]
Let us fix $s>0$.

Notice that thanks to \autoref{f_flin}

\begin{align}\label{ultf1}
\begin{aligned}
\de_{\w}^{(s)}&\left(f(\Theta_{t},\ov{X}),\mathcal{N}(\mu_{t}(\ov{X}),\ov{\K}_{t})\right)\\
&\qquad\leq \de_{\w}^{(s)}\left(f(\Theta_{t},\ov{X}),f^{\lin}(\Theta_{t}^{\lin},\ov{X})\right)+ \de_{\w}^{(s)}\left(f^{\lin}(\Theta_{t}^{\lin},\ov{X}),\mathcal{N}(\mu_{t}(\ov{X}),\ov{\K}_{t})\right)\\
&\qquad\leq 132n^{2}\sqrt{\ov{N}} \left(\norma{Y}_{2}+\sqrt{2n\sqrt{N(m)}}\right)^{2}\left(1+\frac{27}{\left(\lambda_{\min}^{K}\right)^{3}}\right)
\frac{L^2m^2|\mathcal{M}|^5|\mathcal{N}|^2}{N^5(m)}\left(1+\log N(m)\right)\\
&\qquad\quad+\frac{s}{\sqrt{N(m)}} + \de_{\w}^{(s)}\left(f^{\lin}(\Theta_{t}^{\lin},\ov{X}),\mathcal{N}(\mu_{t}(\ov{X}),\ov{\K}_{t})\right).
\end{aligned}
\end{align}
On the other hand, by \autoref{thmwass2}, we obtain that
\begin{align}\label{ultf2}
\begin{aligned}
  \de_{\w}^{(s)}&\left(f^{\lin}(\Theta_{t}^{\lin},\ov{X}),\mathcal{N}(\mu_{t}(\ov{X}),\ov{\K}_{t})\right)\\
  &\qquad\leq C_{m}^{\ov{N},\ov{\K}_{0}}+\frac{\norma{K(\ov{X},X^{T})}_{{\rm op}}}{\lambda_{\min}^{K}}C_{m}^{n,\ov{\K}_{0}}\\
&\qquad\quad+2\left(\frac{1}{\sqrt{N(m)}\lambda_{\min}^{K}} +\frac{6}{\sqrt{N(m)}\left(\lambda_{\min}^{K}\right)^{2}}\sqrt{\ov{N} n}\frac{Lm|\mathcal{M}|^2}{N^2(m)}\right)\mathbb{E}[\|F(0)-Y\|_2]\\
&\qquad\quad+\exp\left(-\frac{1}{256}\frac{N^{4}(m)}{Lm\vert \mathcal{M}\vert^{2} \vert\mathcal{N} \vert^{2}}\frac{1}{N(m)\ov{N}}\right)s
\end{aligned}
\end{align}
Hence, combining \eqref{ultf1}, and \eqref{ultf2} we get

\begin{align}
\begin{aligned}\label{eq:tobesimplified}
\de_{\w}^{(s)}&\left(f(\Theta_{t},\ov{X}),\mathcal{N}(\mu_{t}(\ov{X}),\ov{\K}_{t})\right)\\
&\qquad\leq \frac{s}{\sqrt{N(m)}} + C_{m}^{\ov{N},\ov{\K}_{0}}+\frac{\norma{K(\ov{X},X^{T})}_{{\rm op}}}{\lambda_{\min}^{K}}C_{m}^{n,\ov{\K}_{0}}+\\
&\qquad\quad +132n^{2}\sqrt{\ov{N}} \left(\|Y\|_2+\sqrt{2n\sqrt{N(m)}}\right)^{2}\left(1+\frac{27}{\left(\lambda_{\min}^{K}\right)^{3}}\right)
\frac{L^2m^2|\mathcal{M}|^5|\mathcal{N}|^2}{N^5(m)}\left(1+\log N(m)\right)\\
&\qquad\quad +2\left(\frac{1}{\sqrt{N(m)}\lambda_{\min}^{K}} +\frac{6}{\left(\lambda_{\min}^{K}\right)^{2}}\sqrt{\ov{N}n}\frac{Lm|\mathcal{M}|^2}{N^{5/2}(m)}\right)\mathbb{E}[\|F(0)-Y\|_2]\\
&\qquad\quad+\exp\left(-\frac{1}{256}\frac{N^{4}(m)}{Lm\vert \mathcal{M}\vert^{2} \vert\mathcal{N} \vert^{2}}\frac{1}{N(m)\ov{N}}\right)s.
\end{aligned}
\end{align}
 We recall that
\begin{align}
\begin{aligned}
    C_{m}^{n,\ov{\K}_{0}}\coloneqq  \sqrt{n}\norma{\ov{\K}_{0}^{-1}}_{{\rm op}}\norma{\ov{\K}_{0}}_{{\rm op}}^{\frac{1}{2}} \sqrt{\frac{\widetilde{D}(2nD)^{2}}{(N(m))^{4}}}+\frac{Dmn}{(N(m))^{2}}\omega\left(n,\frac{D}{N(m)}\right),
\end{aligned}
\end{align}
where we recall that
\begin{align}
\begin{aligned}
\omega(d,x)\coloneqq 
\begin{cases}
x\left(C(1,d)-2\log\,x\right) &\text{if $ x \leq 1$,}\\
C(1,d) &\text{if $x>1$.}
\end{cases}
\qquad\text{with}\qquad 
 C(1,d)\coloneqq 2^{\frac{3}{2}}\frac{1+2d}{d}\frac{\Gamma\left(\frac{1+d}{2}\right)}{\Gamma\left(\frac{d}{2}\right)}. 
\end{aligned}
\end{align}
We notice that by proceeding as in \eqref{fcostantona}, we also obtain that
\begin{align}
\begin{aligned}
    C_{m}^{n,\ov{\K}_{0}}&\leq n\sqrt{n}\left(2\norma{\ov{\K}_{0}^{-1}}_{{\rm op}}\norma{\ov{\K}_{0}}_{{\rm op}}^{\frac{1}{2}}+6\right)\frac{m|\mathcal{M}|^{7/2}|\mathcal{N}|^{7/2}}{N^3(m)}\left(1+\log N(m)\right).
\end{aligned}
\end{align}
Furthermore, in order to simplify \eqref{eq:tobesimplified}, we recall \eqref{explitcb}
\begin{align}
    \|K(\bar X,X^T)\|_{\rm op}\leq 4 \sqrt{\ov{N} n}\frac{Lm|\mathcal{M}|^2}{N^2(m)},
\end{align}
we upper bound
\begin{align}
    \left(\|Y\|_2+\sqrt{2n\sqrt{N(m)}}\right)^{2}\leq 2\|Y\|_2^2+4n\sqrt{N(m)}\leq \left(2\|Y\|_2^2+4n\right)\sqrt{N(m)}
\end{align}
and, by virtue of \autoref{A2}, we compute
\begin{align}
    \begin{aligned}
        \mathbb{E}[\|F(0)-Y\|_2]^2&\leq \mathbb{E}[\|F(0)-Y\|_2^2]\\
        &=\sum_{i=1}^n \mathbb{E}\left[\left(f(\Theta_0,x^{(i)})-y^{(i)}\right)^2\right]\\
        &=\sum_{i=1}^n\left( \mathcal{K}_0(x^{(i)},x^{(i)})+(y^{(i)})^2\right)\\
        &\leq n + \|Y\|_2^2,
    \end{aligned}
\end{align}
whence
\begin{align}
    \mathbb{E}[\|F(0)-Y\|_2]&\leq \sqrt{n + \|Y\|_2^2}\leq \sqrt{n}+\|Y\|_2.
\end{align}
Now, we know that\begin{align}
    xe^{-x}\leq \frac{1}{e}
\end{align}
whence
\begin{align}
    e^{-x}\leq \frac{1}{e}\,\frac{1}{x}
\end{align}
holds. Therefore,
\begin{align}
    \exp\left(-\frac{1}{256}\frac{N^{3}(m)}{Lm\vert \mathcal{M}\vert^{2} \vert\mathcal{N} \vert^{2}}\frac{1}{\ov{N}}\right)\leq \frac{256}{e}\ov{N}\,\frac{Lm\vert \mathcal{M}\vert^{2} \vert\mathcal{N} \vert^{2}}{N^{3}(m)}.
\end{align}
The last ingredient is again \eqref{eq:inequalityNm}:
\begin{align}
    N(m)\leq \sqrt{m|\mathcal{M}||\mathcal{N}|}.
\end{align}
Then, it's laborious but easy to show that
\begin{align}
    \begin{aligned}
    \de_{\w}^{(s)}&\left(f(\Theta_{t},\ov{X}),\mathcal{N}(\mu_{t}(\ov{X}),\ov{\K}_{t})\right)\\
    &\qquad\leq \frac{s}{\sqrt{N(m)}}\\
    &\qquad\quad +\ov{N}\sqrt{\ov{N}}\left(2\norma{\ov{\K}_{0}^{-1}}_{{\rm op}}\norma{\ov{\K}_{0}}_{{\rm op}}^{\frac{1}{2}}+6\right)\frac{m|\mathcal{M}|^{7/2}|\mathcal{N}|^{7/2}}{N^3(m)}\left(1+\log N(m)\right)\\
    &\qquad\quad + 4 \sqrt{\ov{N} n}\frac{Lm|\mathcal{M}|^2}{N^2(m)}n\sqrt{n}\left(2\norma{\ov{\K}_{0}^{-1}}_{{\rm op}}\norma{\ov{\K}_{0}}_{{\rm op}}^{\frac{1}{2}}+6\right)\frac{m|\mathcal{M}|^{7/2}|\mathcal{N}|^{7/2}}{N^3(m)}\left(1+\log N(m)\right)\\
    &\qquad\quad + 132n^{2}\sqrt{\ov{N}} \left(2\|Y\|_2^2+4n\right)\sqrt{N(m)}\,\left(1+\frac{27}{\left(\lambda_{\min}^{K}\right)^{3}}\right)
    \frac{L^2m^2|\mathcal{M}|^5|\mathcal{N}|^2}{N^5(m)}\left(1+\log N(m)\right)\\
    &\qquad\quad +2\left(\frac{1}{\sqrt{N(m)}\lambda_{\min}^{K}} +\frac{6}{\left(\lambda_{\min}^{K}\right)^{2}}\sqrt{\ov{N} n}\frac{Lm|\mathcal{M}|^2}{N^{5/2}(m)}\right)\left(\sqrt n + \|Y\|_2\right)\\
    &\qquad\quad + \frac{256}{e}s\,\ov{N}\,\frac{Lm\vert \mathcal{M}\vert^{2} \vert\mathcal{N} \vert^{2}}{N^{3}(m)}\\
    &\qquad\leq s\,\,\frac{m^{9/4}|\mathcal{M}|^{9/4}|\mathcal{N}|^{9/4}}{N^5(m)}\\
    &\qquad\quad +\ov{N}\sqrt{\ov{N}}\left(2\norma{\ov{\K}_{0}^{-1}}_{{\rm op}}\norma{\ov{\K}_{0}}_{{\rm op}}^{\frac{1}{2}}+6\right)\frac{m^{2}|\mathcal{M}|^{9/2}|\mathcal{N}|^{9/2}}{N^{5}(m)}\left(1+\log N(m)\right)\\
    &\qquad\quad + 4 n^2\sqrt{\ov{N}}\left(2\norma{\ov{\K}_{0}^{-1}}_{{\rm op}}\norma{\ov{\K}_{0}}_{{\rm op}}^{\frac{1}{2}}+6\right)\frac{Lm^2|\mathcal{M}|^{11/2}|\mathcal{N}|^{9/2}}{N^5(m)}\left(1+\log N(m)\right)\\
    &\qquad\quad + 132n^{2}\sqrt{\ov{N}} \left(2\|Y\|_2^2+4n\right)\,\left(1+\frac{27}{\left(\lambda_{\min}^{K}\right)^{3}}\right)
    \frac{L^2m^{9/4}|\mathcal{M}|^{21/4}|\mathcal{N}|^{9/4}}{N^5(m)}\left(1+\log N(m)\right)\\
    &\qquad\quad +2\left(\frac{1}{\lambda_{\min}^{K}}\frac{m^{9/4}|\mathcal{M}|^{9/4}|\mathcal{N}|^{9/4}}{N^5(m)} +\frac{6}{\left(\lambda_{\min}^{K}\right)^{2}}\sqrt{\ov{N}n}\frac{Lm^{9/4}|\mathcal{M}|^{13/4}|\mathcal{N}|^{5/4}}{N^5(m)}\right)\left(\sqrt n + \|Y\|_2\right)\\
    &\qquad\quad + \frac{256}{e} s\,\ov{N}\,\frac{Lm^2\vert \mathcal{M}\vert^{3} \vert\mathcal{N} \vert^{3}}{N^{5}(m)}\\
    &\qquad\leq C(\ov{N}, n, \|Y\|_2, \lambda_{\min}^K, \ov{\K}_{0}, s)  \frac{L^2m^{9/4}|\mathcal{M}|^{11/2}|\mathcal{N}|^{9/2}}{N^5(m)}\left(1+\log N(m)\right)  
    \end{aligned}
\end{align}
where
\begin{align}
\begin{aligned}
    C(\ov{N}, n, \|Y\|_2, \lambda_{\min}^K, \ov{\K}_{0}, s)&\coloneqq 96 s\ov{N} +8\sqrt{\ov{N}}\norma{\ov{\K}_{0}^{-1}}_{{\rm op}}\norma{\ov{\K}_{0}}_{{\rm op}}^{\frac{1}{2}}(\ov{N}+4n^2)\\
    &\quad +132n^{2}\sqrt{\ov{N}} \left(2\|Y\|_2^2+4n\right)\,\left(1+\frac{27}{\left(\lambda_{\min}^{K}\right)^{3}}\right)\\
    &\quad+2\left(\sqrt n + \|Y\|_2\right)\left(\frac{1}{\lambda_{\min}^{K}} +\frac{6}{\left(\lambda_{\min}^{K}\right)^{2}}\sqrt{\ov{N} n}\right).
\end{aligned}
\end{align}

\end{proof}

\section{Conclusions}\label{sec:concl}

We have studied quantum neural networks whose generated function is given by the expectation value of the sum of single-qubit observables on the state generated by a parametric quantum circuit (\autoref{A1}).
Our main result is a quantitative proof of the convergence to a Gaussian process of the probability distribution of the functions generated by such networks in the limit of infinite width for both untrained networks with randomly initialized parameters and trained networks.

First, we have considered the probability distribution of the function generated by a quantum neural network with untrained randomly initialized parameters. In \autoref{thmwass1}, we have provided an upper bound to the Wasserstein distance of order $1$ between such probability distribution and the Gaussian distribution with the same covariance.

Then, we have considered the probability distribution of the function generated by the network during the training dynamics via gradient flow.
In \autoref{thmwass4}, we have proved an upper bound to a truncated version of the Wasserstein distance of order $1$ between such probability distribution and a suitable Gaussian process.
The proof of \autoref{thmwass4} is based on \autoref{newgradfl}, which states that, for large enough width, the generated function can be approximated by truncating to the first order its Taylor series with respect to the parameters of the model around their initialization value, implying that each parameter remains with high probability close to its initialization value and that the training happens in the {\em lazy regime}.

Both upper bounds \autoref{thmwass1} and \autoref{newgradfl} tend to zero in the limit of infinite width, provided that the network does not suffer from barren plateaus.
Therefore, our results provide a quantitative proof of the convergence in distribution to a Gaussian process proved in \cite{girardi2024}.

%The results are obtained under some assumptions: In \autoref{A1}, we defined the function generated by the quantum network as the expectation value of a sum of single qubit observables. In \autoref{A2}, we assume that the covariance matrix at initialization is at most 1. Once defined the empirical neural tangent kernel to write the evolution equations of the parameters and the model during training in the form \eqref{gradeq2}, we assume that the NKT, defined by the expectation of the empirical NKT, is invertible when restricted to the training samples (\autoref{A3}). Similar  technical assumptions are required in \cite{girardi2024} for proving the convergence of quantum neural networks toward Gaussian processes, as we summarize in \autoref{sec:towgauss}.

%Remarkably, this paper extends the recent work \cite{girardi2024}, where the authors proved that in the limit of infinite width, the probability distribution of the function generated by the trained quantum network converges in distribution to a Gaussian process whose mean and covariance can be computed analytically.

While \cite{girardi2024} requires to build a sequence of networks with diverging width, our results can be applied to any quantum neural network with finite width.
Moreover, \cite{girardi2024} proves the convergence in distribution to a Gaussian process only at fixed training time, and such result holds in the limit $t\to\infty$ only if such limit is taken after the limit of infinite width.
Instead, all the upper bounds proved in this work do not depend on the training time and hold in the limit $t\to\infty$ for finite width.

Our results open the way to several possible research directions:
\begin{itemize}
    \item It would be interesting to extend our results to the case of an infinite input space, such as a subset of $\mathbb{R}^d$. In particular, we would like to generalize the functional bounds for the Wasserstein distance for classical neural networks, see for instance \cite[Theorem 1.1]{balasub2024}, and \cite[Theorem 3.12]{favaro2024q}.
    \item It would also be interesting to determine lower bounds to the distance between the probability distribution of the function generated by the trained network and the corresponding Gaussian process, or at least to determine the asymptotic scaling of such corrections in the limit of infinite width, providing a quantum generalization of the classical results in this direction \cite{hanin2019finitedepthwidthcorrections,yaida2020nongaussian,Roberts_2022,hanin2023randomfullyconnectedneural}.
\end{itemize}

\section*{Acknowledgements}
GDP has been supported by the HPC Italian National Centre for HPC, Big Data and Quantum Computing -- Proposal code CN00000013 -- CUP J33C22001170001 and by the Italian Extended Partnership PE01 -- FAIR Future Artificial Intelligence Research -- Proposal code PE00000013 -- CUP J33C22002830006 under the MUR National Recovery and Resilience Plan funded by the European Union -- NextGenerationEU.
Funded by the European Union -- NextGenerationEU under the National Recovery and Resilience Plan (PNRR) -- Mission 4 Education and research -- Component 2 From research to business -- Investment 1.1 Notice Prin 2022 -- DD N. 104 del 2/2/2022, from title ``understanding the LEarning process of QUantum Neural networks (LeQun)'', proposal code 2022WHZ5XH -- CUP J53D23003890006.
DP has been supported by project SERICS (PE00000014) under the MUR National Recovery and Resilience Plan funded by the European Union -- NextGenerationEU.
GDP and DP are members of the ``Gruppo Nazionale per la Fisica Matematica (GNFM)'' of the ``Istituto Nazionale di Alta Matematica ``Francesco Severi'' (INdAM)''. AMH has been supported by project PRIN 2022 
``understanding the LEarning process of QUantum Neural networks (LeQun)'', proposal code 2022WHZ5XH -- CUP J53D23003890006. The author AMH is a member of the ``Gruppo Nazionale per l'Analisi Matematica, la Probabilità e le loro Applicazioni (GNAMPA)'' of the ``Istituto Nazionale di Alta Matematica ``Francesco Severi'' (INdAM)''.

\appendix

\section{Regularity of the solutions of the Stein equation}\label{sec:regularity}
In what follows, we recall some recent result about the regularity of the solutions of the Stein equation \eqref{pdestein1} obtained in \cite{gall2018}. Let us fix $d\geq 1$, and suppose that $\Sigma= \I_{d\times d}$ is identity matrix in $\R^{d}$. Given $h:\R^{d}\rightarrow \R$ a $1$-Lipschitz function, we consider $f_{h}$ as a solution of the Stein equation \eqref{pdestein1}.   

\begin{prop}\label{propreg1}
Let $h:\R^{d}\rightarrow \R$ be a $1$-Lipschitz function. Then the solution $f_{h}$ of the Stein equation \eqref{pdestein1} satisfies,
\begin{align}\label{regformula}
 \norma{{\rm Hessian} f_{h}(x)- {\rm Hessian} f_{h}(y)}_{2} \leq 
 \displaystyle\begin{cases}
 \norma{x-y}_{2}\left(C(1,d)-2\log(\norma{x-y}_{2})\right), & \text{if $\norma{x-y}_{2}\leq 1$;}\\
C(1,d) & \text{if $\norma{x-y}_{2}>1$}.
 \end{cases}
\end{align}
Here, we have denoted by $C(1,d)$ the constant given by
\begin{align}\label{costreg}
 C(1,d)\coloneqq 2^{\frac{3}{2}}\frac{1+2d}{d}\frac{\Gamma\left(\frac{1+d}{2}\right)}{\Gamma\left(\frac{d}{2}\right)}.  
\end{align}
\end{prop}
\begin{proof}
This result was proved in \cite[Proposition 2.3]{gall2018} in a more general case.    
\end{proof}
In the next, let us define $\omega(d,\cdot): [0,\infty]\rightarrow [0,\infty]$ the following function:
\begin{align}\label{contfunct}
\begin{aligned}
\omega(d,x)\coloneqq 
\begin{cases}
x\left(C(1,d)-2\log\,x\right), & \text{if $ x \leq 1$,}\\
C(1,d), & \text{if $x>1$.}
\end{cases}
\end{aligned}
\end{align}
Notice that $\omega(d,0)=0$, $\lim_{x\rightarrow 0^{+}}\omega(d;x)=0$, and it is increasing, then it is a modulus of continuity.

\section{Convergence Towards a Gaussian Process}\label{sec:towgauss}

In this Appendix, we briefly summarize some of the key results established in \cite{girardi2024}.
There, the authors consider a sequence of quantum neural networks with diverging width satisfying \autoref{A1}-\autoref{A4} and the following further assumptions: 
\begin{ass}\label{H2}
Suppose that the covariance of the generated function at initialization converges uniformly:
\begin{align}\label{ipot2}
    \lim_{m\rightarrow +\infty}\sup_{x,x'\in \X}\left\vert\E\left(f(\Theta,x)f(\Theta,x')\right) - \K_0^\infty(x,x')\right\vert =0\,,
\end{align}
where $\K_0^\infty:\X\times\X\rightarrow \R$ is a kernel such that 
\begin{align}\label{ipot3}
    \K_0^\infty(x,x)>0 \hskip 0,1cm \text{for all $x\in\X$.}
\end{align}
\end{ass}
\begin{ass}\label{H3}
Suppose that $N(m)$ grows sufficiently fast such that
\begin{align}
\lim_{m\rightarrow +\infty}\frac{m \vert \mathcal{M}\vert^{2}\vert\mathcal{N}\vert^{2}}{N^{3}(m)}=0.
\end{align}
\end{ass}
\begin{ass}\label{lasthypconv}
Alternatively, suppose that
\begin{align}
\lim_{m\rightarrow +\infty}\frac{L^{2}m^{2}\vert\mathcal{M}\vert^{6}\vert\mathcal{N}\vert^{3}}{N^5(m)}\log\,N(m)=0. 
\end{align}    
\end{ass}
\begin{ass}\label{H4}
Suppose that there exist a  positive constant $N_{K}(m)$ depending on the number of qubits $m$ and a kernel $K^\infty:\X\times \X \rightarrow \R$ such that the analytic neural tangent kernel rescaled by $N_K(m)$ converges uniformly to $K^\infty$:
\begin{align}\label{limitk}
\lim_{m\rightarrow +\infty}\sup_{x,x'\in \X}\left\vert \frac{K(x,x')}{N_{K}(m)}-K^\infty(x,x')\right\vert=0.    
\end{align}
 Further, suppose that the minimum eigenvalue of $K^\infty= K^\infty(X,X^{T})$ is strictly positive, and denote it by $\lambda_{{\rm min}}^{\infty}$.
\end{ass} 
The following is one of the main achievements of \cite{girardi2024}.
\begin{thm}[{\cite[Theorem 3.14]{girardi2024}}]\label{convergencethm}
Suppose that \autoref{A1}-\autoref{A2}, \autoref{H2}-\autoref{H3} hold true. Then as $m\rightarrow +\infty$, the probability distribution of the generated function at initialization converges in distribution to the Gaussian process with mean zero and covariance $\K_{0}^\infty$ satisfying our assumption \autoref{H2}.
\end{thm}
We notice that \autoref{convergencethm} does not provide any convergence rate. Next, we prove that \autoref{convergencethm} can be derived from Theorem \autoref{thmwass1}.
\begin{cor}
Let us suppose that \autoref{A1}, \autoref{A2}, and 

\begin{align}\label{newassumpt}
 \lim_{m\rightarrow +\infty}\frac{m|\mathcal{M}|^{7/2}|\mathcal{N}|^{7/2}}{N^3(m)}\left(1+\log N(m)\right)=0. 
 \end{align}
Then 
\begin{align}
    \lim_{m\rightarrow +\infty}\de_{\w}\left(f(\Theta,\ov{X}),\mathcal{N}(0,\K_{0}^\infty)\right)=0,
\end{align}
so that $x\mapsto f(\Theta,x)$ converges in distribution to a Gaussian process with mean zero, and covariance operator $\K_{0}^\infty$.
\end{cor}
\begin{proof}
Follows from \autoref{thmwass1} and \eqref{newassumpt}.
\end{proof}
Let us notice that to establish the convergence result \autoref{convergencethm}, in \cite{girardi2024}, the authors used the Lévy's continuity theorem combined with an appropriate upper bound on the cumulants of $f$. This approach was necessary because the random variables $\{f_{k}\}_{k=1}^{m}$ defined in \eqref{model1} exhibit weak dependence, meaning the standard central limit theorem cannot be directly applied \cite{feray2016,janson1988}. Nevertheless, they were able to impose the weaker convergence assumption \autoref{H3} instead of \eqref{newassumpt}. 
Next, let us recall another remarkable convergence result proved in \cite{girardi2024} regarding the behavior of the trained objective function $f(\Theta_{t},\cdot)$.

\begin{thm}\label{lasthm}
Let us consider a sequence of quantum neural networks trained via gradient flow with learning rate $\frac{\eta}{N_K(m)}$
\begin{align}\label{gradeq1bis}
 \displaystyle\begin{cases}
  \frac{\de \Theta_{t}}{\de t}=-\frac{\eta}{N_K(m)}\nabla_{\Theta}f(\Theta_{t},X^{T})\nabla_{f(\Theta_{t},X)}\L(\Theta_{t}),\\
 \frac{\de}{\de t}f(\Theta_{t},x)=-\frac{\eta}{N_K(m)}\left(\nabla_{\Theta}f(\Theta_{t},x)\right)^{T}\nabla_{\Theta}f(\Theta_{t},X^{T})\nabla_{f(\Theta_{t},X)}\L(\Theta_{t})\,,
 \end{cases}   
\end{align}
and satisfying \autoref{A1}, \autoref{A2}, \autoref{H2}, \autoref{lasthypconv}, \autoref{H4}. Then, for any $t\geq 0$, in the limit of infinitely many qubits $m\rightarrow +\infty$, $\{f(\Theta_{t},x)\}_{x\in \X}$ converges in distribution to a Gaussian process $\{f_{t}^{(\infty)}(x)\}_{x\in\X}$ with mean and covariance 

\begin{align}
&\begin{aligned}\label{KK}
\K^\infty_{t}(x,x') &\coloneqq \K^\infty_{0}(x,x')-K^\infty(x,X^{T})(K^\infty)^{-1}\left(\mathbbm{1}-e^{-t\eta\,K^\infty}\right)\mathcal{K}^\infty_{0}(X,x')\\
&-K^\infty(x',X^{T})(K^\infty)^{-1}\left(\mathbbm{1}-e^{-t\eta\,K^\infty}\right)\mathcal{K}^\infty_{0}(X,x)+\\
&+K^\infty(x,X^{T})(K^\infty)^{-1}\left(\mathbbm{1}-e^{-t\eta\,K^\infty}\right)\mathcal{K}^\infty_{0}(X,X^{T})\left(\mathbbm{1}-e^{-t\eta\,K^\infty}\right)(K^\infty)^{-1}K^\infty(X,x'),
\end{aligned}\\\label{MM}
&\begin{aligned}
\quad\mu^\infty_{t}(x) &\coloneqq K^\infty(x,X^{T})(K^\infty)^{-1}\left(\mathbbm{1}-e^{-t\eta\,K^\infty}\right)Y,    
\end{aligned}
\end{align}
where $K^\infty\coloneq K^\infty(X,X^T)$.
\end{thm}
%Let us provide some comments about the hypotheses of the previous convergence result. In assumption \autoref{H4}, the existence of the normalizing constant $N_{K}(m)$ is assumed to effectively rescale the analytic kernel, ensuring that in the limit of infinitely many qubits, a well-defined limit for the trained quantum neural function is obtained. By contrast, our study has primarily focused on quantitatively analyzing the behavior of the trained quantum objective function $f(\Theta_{t},\cdot)$ in relation to its corresponding Gaussian process, which is driven by the covariance operator defined \eqref{covlimit}, and the mean \eqref{newmedia}.
Next, we prove that the conclusion of \autoref{lasthm} can be obtained from \autoref{thmwass4}.

\begin{cor}
Let us consider a sequence of quantum neural networks trained via the gradient flow \eqref{gradeq1bis} and asymptotically satisfying \autoref{A1}, \autoref{A2}, \autoref{H2}, \autoref{H4}, and

\begin{align}\label{furtherssumpt1}
\lim_{m\rightarrow +\infty}\frac{L^2m^{9/4}|\mathcal{M}|^{11/2}|\mathcal{N}|^{9/2}}{N^5(m)}\left(1+\log N(m)\right)=0.
\end{align}
Then, for any $s>0$ and any $t>0$, we have   

\begin{align}
\lim_{m\rightarrow +\infty}\de_{\w}^{(s)}\left(f(\Theta_{t},\ov{X}),\mathcal{N}(\mu^\infty_{t}(\ov{X}),\K^\infty_{t})\right)=0,  
\end{align}
so that the trained quantum neural function $x\mapsto f(\Theta_{t},x)$ converges in distribution to a Gaussian process with mean $\mu^\infty_{t}$, and covariance operator $\K^\infty_{t}$ defined by \eqref{MM}, and \eqref{KK}, respectively.
\end{cor}
\begin{proof}
\autoref{thmwass4} provides a bound on the truncated $W_1$ distance between the function generated by the trained network at finite width and the Gaussian process with mean and covariance
\begin{align}
&\begin{aligned}
\widetilde{\K}_{t}(x,x')\coloneqq &\K_{0}(x,x')-\widetilde{K}(x,X^{T})\widetilde{K}^{-1}\left(\mathbbm{1}-e^{-t\eta\,\widetilde{K}}\right)\mathcal{K}_{0}(X,x')\\
&-\widetilde{K}(x',X^{T})\widetilde{K}^{-1}\left(\mathbbm{1}-e^{-t\eta\,\widetilde{K}}\right)\mathcal{K}_{0}(X,x)+\\
&+\widetilde{K}(x,X^{T})\widetilde{K}^{-1}\left(\mathbbm{1}-e^{-t\eta\,\widetilde{K}}\right)\mathcal{K}_{0}(X,X^{T})\left(\mathbbm{1}-e^{-t\eta\,\widetilde{K}}\right)\widetilde{K}^{-1}K(X,x'),
\end{aligned}\\
&\widetilde{\mu}_{t}(x)\coloneqq \widetilde{K}(x,X^{T})\widetilde{K}^{-1}\left(\mathbbm{1}-e^{-t\eta\,\widetilde{K}}\right)Y,
\end{align}
where
\begin{equation}
    \widetilde{K}(x,x')\coloneqq \frac{K(x,x')}{N_{K}(m)}\qquad\forall\;x,\,x'\in\mathcal{X}\,.
\end{equation}
Notice that by the triangle inequality, we have that
\begin{align}\label{exit2}
\de_{\w}^{(s)}\left(f(\Theta_{t},\ov{X}),\mathcal{N}(\mu^\infty_{t}(\ov{X}),\K^\infty_{t})\right)\leq  \de_{\w}^{(s)}\left(f(\Theta_{t},\ov{X}),\mathcal{N}(\tilde{\mu}_{t}(\ov{X}),\tilde{\K}_{t})\right) + \de_{\w}^{(s)}\left(\mathcal{N}(\tilde{\mu}_{t}(\ov{X}),\tilde{\K}_{t}),\mathcal{N}(\mu^\infty_{t}(\ov{X}),\K^\infty_{t})\right). 
\end{align}
By \autoref{thmwass4}, and by \cite[Remark 4.11]{girardi2024}, we simply let $N_{K}(m)\geq 1$, and then we still have that

\begin{align}\label{exit3}
\de_{\w}^{(s)}\left(f(\Theta_{t},\ov{X}),\mathcal{N}(\widetilde{\mu}_{t}(\ov{X}),\widetilde{\K}_{t})\right)
&\leq  C(\ov{N}, n, \|Y\|_2, \lambda_{\min}^K, \mathcal{K}_0, s)  \frac{L^2m^{9/4}|\mathcal{M}|^{11/2}|\mathcal{N}|^{9/2}}{N^5(m)}\left(1+\log N(m)\right)\,,
\end{align}
where $C(\ov{N}, n, \|Y\|_2, \lambda_{\min}^K, \ov{\K}_{0}, s)$ is the positive constant given by \eqref{newcostante}. Then by \eqref{furtherssumpt1} we obtain that \eqref{exit3} tends to zero as $m$ goes to infinity. On the other hand, notice that

\begin{align}
\de_{\w}^{(s)}\left(\mathcal{N}(\mu_{t}(\ov{X}),\K_{t}),\mathcal{N}(\mu^\infty_{t}(\ov{X}),\K^\infty_{t})\right)\leq \de_{\w_{2}}\left(\mathcal{N}(\mu_{t}(\ov{X}),\K_{t},\mathcal{N}(\mu^\infty_{t}(\ov{X}),\K^\infty_{t})\right)
\end{align}
where $d_{\w_{2}}$ is the Wasserstein distance of order $2$. Moreover, by \cite[Equation (4)]{dowson1982}

\begin{align}
\left(\de_{\w_{2}}\left(\mathcal{N}(\widetilde{\mu}_{t}(\ov{X}),\widetilde{\K}_{t},\mathcal{N}(\mu^\infty_{t}(\ov{X}),\K^\infty_{t})\right)\right)^{2}=\norma{\widetilde{\mu}_{t}(\ov{X})- \mu^\infty_{t}(\ov{X})}_{2}^2+ \tr\left(\widetilde{\K}_{t}+\K^\infty_{t}-2(\widetilde{\K}_{t}\K^\infty_{t})^{\frac{1}{2}}\right).   
\end{align}
Then, by applying \autoref{H4}, we conclude that

\begin{align}\label{exit4}
\lim_{m\rightarrow +\infty}\de_{\w_{2}}\left(\mathcal{N}(\widetilde{\mu}_{t}(\ov{X}),\widetilde{\K}_{t}),\mathcal{N}(\mu^\infty_{t}(\ov{X}),\K^\infty_{t})\right)=0.   
\end{align}
Therefore, by combining \eqref{newassumpt} with \eqref{exit2}, \eqref{exit3}, \eqref{exit4}, we obtain our desired convergence.

\end{proof}

\bibliographystyle{unsrt}
\bibliography{bibliography}

\end{document}